\documentclass[preprint,3p,11pt]{elsarticle}

\usepackage{amssymb}
\usepackage{amsthm}
\usepackage{graphicx}
\usepackage{textcomp}
\usepackage{mathtools}
\usepackage{color,soul,xcolor}
\usepackage{url}
\usepackage{algpseudocode}
\usepackage{algorithm}
\usepackage{outlines}
\usepackage{comment}
\usepackage{microtype}
\usepackage{textgreek}

\usepackage{tikz}
\usepackage{pgfplots}
\usepackage{grffile}
\pgfplotsset{compat=newest}
\usetikzlibrary{automata}
\usetikzlibrary{arrows, arrows.meta, decorations.markings, patterns}
\usepgfplotslibrary{patchplots, groupplots}
\usetikzlibrary{shapes, positioning, fit, backgrounds}
\usepgfplotslibrary{fillbetween}

\usetikzlibrary{overlay-beamer-styles}
\usetikzlibrary{calc}

\usepackage{calc}

\definecolor{tudcyan}{RGB}{0,166,214}
\definecolor{tudmagenta}{RGB}{109,23,127}
\definecolor{tudpurple}{RGB}{29,28,115}
\definecolor{tudgraygreen}{RGB}{107,134,137}

\colorlet{lighttudcyan}{tudcyan!20}
\colorlet{lighttudmagenta}{tudmagenta!20}

\newlength{\hatchspread}
\newlength{\hatchthickness}
\newlength{\hatchshift}
\newcommand{\hatchcolor}{}
\tikzset{hatchspread/.code={\setlength{\hatchspread}{#1}},
	hatchthickness/.code={\setlength{\hatchthickness}{#1}},
	hatchshift/.code={\setlength{\hatchshift}{#1}},
	hatchcolor/.code={\renewcommand{\hatchcolor}{#1}}}
\tikzset{hatchspread=7pt,
	hatchthickness=0.5pt,
	hatchshift=0pt,
	hatchcolor=black}
\pgfdeclarepatternformonly[\hatchspread,\hatchthickness,\hatchshift,\hatchcolor]
{custom north west lines}
{\pgfqpoint{\dimexpr-2\hatchthickness}{\dimexpr-2\hatchthickness}}
{\pgfqpoint{\dimexpr\hatchspread+2\hatchthickness}{\dimexpr\hatchspread+2\hatchthickness}}
{\pgfqpoint{\dimexpr\hatchspread}{\dimexpr\hatchspread}}
{
	\pgfsetlinewidth{\hatchthickness}
	\pgfpathmoveto{\pgfqpoint{0pt}{\dimexpr\hatchspread+\hatchshift}}
	\pgfpathlineto{\pgfqpoint{\dimexpr\hatchspread+0.15pt+\hatchshift}{-0.15pt}}
	\ifdim \hatchshift > 0pt
	\pgfpathmoveto{\pgfqpoint{0pt}{\hatchshift}}
	\pgfpathlineto{\pgfqpoint{\dimexpr0.15pt+\hatchshift}{-0.15pt}}
	\fi
	\pgfsetstrokecolor{\hatchcolor}
	\pgfusepath{stroke}
}

\def\centerarc[#1](#2)(#3:#4:#5)
{ \draw[#1] ($(#2)+({#5*cos(#3)},{#5*sin(#3)})$) arc (#3:#4:#5); }
\newcommand*{\tran}{^{\mkern-1.5mu\mathsf{T}}\!}  

\def\d{\ensuremath{\mathrm{d}}}
\def\imag{\ensuremath{\mathrm{i}}}

\DeclareMathOperator{\LimAvg}{LimAvg}
\DeclareMathOperator{\Avg}{Avg}
\DeclareMathOperator{\SAC}{SAC}
\DeclareMathOperator{\VS}{\underbar{V}}
\DeclareMathOperator{\VL}{\overline{V}}
\DeclareMathOperator{\Id}{Id}

\DeclarePairedDelimiter{\floor}{\lfloor}{\rfloor}
\DeclareMathOperator{\bigO}{\mathcal{O}}

\def\norm[#1]{\left|#1\right|}
\def\shortnorm[#1]{|#1|}
\def\nd{\ensuremath{d_{\mathrm{n}}}}
\def\bisim{\ensuremath{\cong}}
\DeclareMathOperator{\cl}{cl}

\def\dummy{}

\def\Xs{\mathcal{X}}
\def\Ys{\mathcal{Y}}
\def\Us{\mathcal{U}}

\def\No{\mathbb{N}_{0}}
\def\N{\mathbb{N}}
\def\R{\mathbb{R}}
\def\S{\mathbb{S}}
\def\Q{\mathbb{Q}}
\def\C{\mathbb{C}}

\def\Rs{\mathcal{R}}

\def\Ss{\mathcal{S}}
\def\Ks{\mathcal{K}}
\def\Qs{\mathcal{Q}}

\def\As{\mathcal{A}}  
\def\Es{\mathcal{E}}  
\def\Bs{\mathcal{B}}  
\def\Us{\mathcal{U}}  
\def\Ps{\mathcal{P}}  

\def\nup{{n_{\mathrm{u}}}}
\def\nx{{n_{\mathrm{x}}}}

\def\xv{\boldsymbol{x}}

\def\wv{\boldsymbol{w}}
\def\vv{\boldsymbol{v}}
\def\xiv{\boldsymbol{\xi}}
\def\upsilonv{\boldsymbol{\upsilon}}

\def\lv{\boldsymbol{l}}

\def\Am{\boldsymbol{A}}
\def\Bm{\boldsymbol{B}}

\def\I{{\normalfont\textbf{I}}}
\def\O{{\normalfont\textbf{0}}}

\def\Mm{\boldsymbol{M}}
\def\Nm{\boldsymbol{N}}

\def\Pm{\boldsymbol{P}}
\def\Qm{\boldsymbol{Q}}

\def\Vm{\boldsymbol{V}}

\def\Km{\boldsymbol{K}}

\def\e{\mathrm{e}}
\def\imag{\ensuremath{\mathrm{i}}}
\def\piconst{\textup{\textpi}}
\def\uppi{{\mathrm{\pi}}}

\newtheorem{defn}{Definition}
\newtheorem{rem}{Remark}
\newtheorem{prop}{Proposition}
\newtheorem{lem}{Lemma}
\newtheorem{thm}{Theorem}

\newtheorem{ex}{Example}

\newcommand{\fakeparagraphnospace}[1]{\vspace{1mm}\noindent\textbf{#1.}}
\newcommand{\fakeparagraph}[1]{\fakeparagraphnospace{#1}}

\begin{document}

\begin{frontmatter}



\title{Computing the average inter-sample time of event-triggered control using quantitative automata\tnoteref{t1}}


\author{Gabriel de Albuquerque Gleizer}
\author{Manuel Mazo Jr.}

\address{TU Delft, Mekelweg 2, Delft, 2628 CD, ZH, The Netherlands.}

\tnotetext[t1]{This work is supported by the European Research Council through the SENTIENT project [ERC-2017-STG \#755953].}
\begin{abstract}
Event-triggered control (ETC) is a major recent development in cyber-physical systems due to its capability of reducing resource utilization in networked devices. However, while most of the ETC literature reports simulations indicating massive reductions in the sampling required for control, no method so far has been capable of quantifying these results.
In this work, we propose an approach through finite-state abstractions to do formal quantification of the traffic generated by ETC of linear systems, in particular aiming at computing its smallest average inter-sample time (SAIST).
The method involves abstracting the traffic model through $l$-complete abstractions, finding the cycle of minimum average length in the graph associated to it, and verifying whether this cycle is an infinitely recurring traffic pattern. The method is proven to be robust to sufficiently small model uncertainties, which allows its application to compute the SAIST of ETC of nonlinear systems.

\end{abstract}

\begin{keyword}
event-triggered control \sep hybrid systems \sep abstractions

\end{keyword}

\end{frontmatter}


\section{INTRODUCTION}

In modern control applications, smart sensors, controllers, and actuators communicate with each other through digital communication networks. The standard networked control approach is periodic \emph{sample-and-hold} control: at every $h$ time units, sensors sample their values, send them through the network to the controller, which then updates its control command to the actuators; the command is held constant in between samples. Obviously, small values of the sampling period $h$ approximate the control performance to that of the idealized continuous controller, but increase bandwidth usage and radio energy consumption in wireless networks. This single parameter therefore limits the size and applicability of networked control systems (NSCs), and a natural question that has arisen is how to design aperiodic sampling approaches. In \cite{aastrom2002comparison}, the idea of sampling based on an event --- the error between the current state and the last sampled state exceeding a threshold --- was investigated with the name of Lebesgue sampling (after the Lebesgue integration). This idea was further developed in \cite{tabuada2007event}, where for the first time a framework for asymptotic stabilization of the origin through an event-based sampling was conceived. This approach is now known as \emph{event-triggered control} (ETC), and, given the enormous reductions in sampling it showed in early simulations, immense interested followed. Significant focus was given on event design to reduce sampling frequency while guaranteeing stability and control performance (e.g.~\cite{wang2008event, girard2015dynamic, heemels2012introduction}), extend ETC to different control structures \cite{heemels2012introduction}, or improve practical implementation aspects of ETC, such as the \emph{periodic event-triggered control} (PETC) of \cite{heemels2013periodic}, where event conditions are checked periodically. 
It is remarkable, however, that until very recently \cite{gleizer2021hscc}, no method to formally compute ETC sampling performance existed. Typically, ETC papers limit their formal results to stability, control performance, and Zeno-freeness --- the absence of Zeno behavior, or infinitely fast sampling in finite time. Similarly to Zeno-freeness, in PETC it is immediate that its average sampling is in the worst case the same as a baseline periodic control whose sampling period is the same $h$ as the event checking period of PETC. The critical question is, \emph{how significant are the savings provided by ETC?} This is a \emph{quantitative} question, and as such it requires the computation of sampling performance metrics for ETC. 

As previously mentioned, only recently there has been investigation of ETC traffic patterns, which can be categorized in two main approaches. 
The first category \cite{postoyan2019interevent, rajan2020analysis} focuses on understanding the qualitative asymptotic trends of the \emph{inter-sample times} (ISTs) of planar linear systems. In \cite{postoyan2019interevent}, the authors conclude that, under some conditions, the ISTs eventually converge to a fixed value or exhibit an oscillatory pattern. Despite providing very interesting insights, the results are limited to two-dimensional state spaces, and do not provide the quantitative information that we consider crucial. The second category uses symbolic abstractions \cite{kolarijani2016formal, gleizer2020scalable}, following on the extensive work on state-space partitioning and aggregation for abstractions, see \cite{tabuada2009verification}. In \cite{kolarijani2016formal} and \cite{gleizer2020scalable}, the prediction of ISTs is focused on the scheduling problem: in this context, a scheduler can use finite-state \emph{traffic models} to request sensor data before events are triggered in order to prevent collisions. However, these traffic models do not capture effectively long-term traffic properties of ETC, which hampers their use for quantitative analysis. Still in the same category, \cite{gleizer2020towards} uses a bisimulation-like algorithm that determines the $m$ next ISTs from a given state, followed by a very conservative estimate of the worst-case average IST by taking the minimum average of all such $m$-length sequences.

We have recently presented in \cite{gleizer2021hscc} a first approach to compute the average IST of PETC for linear systems, which is an important step in the direction we deem fundamental in order to judge the practical relevance of ETC. More specifically, we devised an algorithm to compute the \emph{smallest average inter-sample time}, or SAIST. This constitutes a natural metric which directly translates into average resource utilization in a network. Our approach in \cite{gleizer2021hscc} is based on the abstraction of a closed-loop PETC system into a weighted finite-state automaton, where the weight of a transition is the IST generated by the state. The smallest-in-average cycle (SAC) of the weighted graph associated with the abstraction is proven to be a lower bound, and in some cases the exact value, of the PETC's SAIST. The exact cases happen when a possible infinite sequence of ISTs of the PETC is the infinite repetition of this SAC. This paper is an extension of \cite{gleizer2021hscc}; here, 
\begin{enumerate}[(i)]
	\item we present a general version of the algorithm for verifying the limit average metric of an infinite-state system, as well as some behavioral conditions for its termination and how to compute uncertainty bounds;
	\item we prove that, in the general case, working with linear invariant subspaces of a linear map is necessary and sufficient to prove that a given SAC can repeat infinitely often as a sequence of ISTs;
	\item we show that the algorithm is robust to small enough model uncertainties --- this enables us to elaborate on the computability of SAIST of linear systems and, moreover, allows the SAIST computation of nonlinear PETC systems;
	\item we provide more numerical examples and their associated conclusions, including how to decrease the required amount of computations for the abstraction.
\end{enumerate}

The more general results rely on a behavioral interpretation of dynamical systems \cite{willems1991paradigms} and the associated abstraction methods \cite{moor1999supervisory, schmuck2015comparing}. The specialized results for PETC SAIST are based on a combination of quotient-based abstractions \cite{tabuada2009verification} and a behavioral-based analysis. Overall, our new results help consolidating the methodology proposed in \cite{gleizer2021hscc}, equipping engineers with a tool to formally estimate the benefits of ETC applications.

This paper follows the following structure: The main problem is stated in §\ref{sec:prob}. Background and preliminary results about (quantitative) abstractions, including the basic results from \cite{gleizer2021hscc}, are shown in §\ref{sec:prem}. Then, a general pseudo-algorithm to compute limit average metrics of infinite systems is presented in §\ref{sec:genericlimavg}, while its specialization for PETC SAIST computation is presented in §\ref{sec:mainpetc}. Finally, numerical examples are given in §\ref{sec:numerical}, and conclusions and future work are discussed in \ref{sec:conc}.

\subsection{Notation}

We denote by $\No$ the set of natural numbers including zero, $\N \coloneqq \No \setminus \{0\}$, $\N_{\leq n} \coloneqq \{1,2,...,n\}$, by $\Q$ the set of rational numbers, and by $\R_+$ the set of non-negative reals. For a complex number $z \in \C, z^*$ denotes its complex conjugate, $\arg{z}$ denotes its argument, and $\Im(z)$ denotes its imaginary part. %
We denote by $|\xv|$ the Euclidean norm of a vector $\xv \in \R^n$ and by $\norm[\Am]$ the 2-induced norm of a matrix $\Am \in \R^{n \times m}$, but if $s$ is a sequence or set, $|s|$ denotes its length or cardinality, respectively. The set $\S^n$ denotes the set of symmetric matrices in $\R^n$. For a symmetric matrix $\Pm \in \S^n,$ we write $\Pm \succ \O$ ($\Pm \succeq \O$) if $\Pm$ is positive definite (semi-definite). %
For a set $\Xs\subseteq\Omega$, we denote by $\cl(\Xs)$ its closure, $\partial\Xs$ its boundary, and $\bar{\Xs}$ its complement $\Omega \setminus \Xs$.
We often use a string notation for sequences, e.g., $\sigma = abc$ reads $\sigma(1) = a, \sigma(2) = b, \sigma(3) = c.$ Powers and concatenations work as expected, e.g., $\sigma^2 = \sigma\sigma = abcabc.$ In particular, $\sigma^\omega$ denotes the infinite repetition of $\sigma$.
For a relation $\Rs \subseteq \Xs_a \times \Xs_b$, its inverse is denoted as $\Rs^{-1} = \{(x_b, x_a) \in \Xs_b \times \Xs_a : (x_a, x_b) \in \Rs\}$. Finally, we denote by $\pi_\Rs(\Xs_a) \coloneqq \{ x_b \in \Xs_b \mid (x_a, x_b) \in \Rs \text{ for some } x_a \in \Xs_b\}$ the natural projection of $\Xs_a$ onto $\Xs_b$.

\section{PROBLEM STATEMENT}\label{sec:prob}

Consider a linear time-invariant plant controlled with sample-and-hold state feedback \cite{aastrom2013computer} %
described by
\begin{equation}\label{eq:plant}
\begin{aligned}
	\dot{\xiv}(t) &= \Am\xiv(t) + \Bm\upsilonv(t), \\
	\upsilonv(t) &= \Km\hat{\xiv}(t),
\end{aligned}
\end{equation}
where $\xiv(t) \in \R^\nx$ is the plant's state with initial value $\xv_0 \coloneqq \xiv(0)$, $\hat{\xiv}(t) \in \R^\nx$ is the state measurement available to the controller, $\upsilonv(t) \in \R^\nup$ is the control input, $\nx$ and $\nup$ are the state-space and input-space dimensions, respectively, and $\Am, \Bm, \Km$ are matrices of appropriate dimensions. 
The measurements are updated to the controller only at specific \emph{sampling times}, with their values being zero-order held on the controller: let $t_i \in \R_+, i \in \N_0$ be a sequence of sampling times, with $t_0 = 0$ and $t_{i+1} - t_i > \varepsilon$ for some $\varepsilon > 0$; then $\hat{\xiv}(t) = \xiv(t_i), \forall t \in [t_i, t_{i+1})$.

In ETC, a \emph{triggering condition} determines the sequence of times $t_i$. In the case of PETC, this condition is checked only periodically, with a fundamental checking period $h$. Throughout this paper, we assume the time units have been scaled so that $h = 1$.%
\footnote{\label{foot:hequals1}This time re-scaling can be achieved by simply multiplying $\Am$ and $\Bm$ with $h$.}
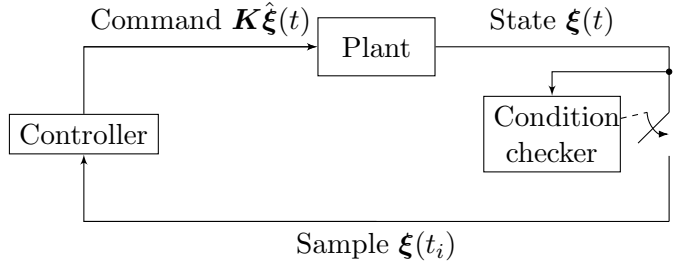
\begin{figure}
	\begin{center}
		\tikzstyle{block} 		= [draw, rectangle, minimum height=2em, minimum width=4em]
\tikzstyle{sum} 		= [draw, circle, inner sep=0, minimum size=0.2cm, node distance=1cm]
\tikzstyle{input} 		= [coordinate]
\tikzstyle{output} 		= [coordinate]
\tikzstyle{split} 		= [coordinate]
\tikzstyle{pinstyle} 	= [pin edge={to-,thin,black}]

\begin{tikzpicture}[auto, node distance=1.5em,>=latex']

\node [block, align=center,] (P) {Plant};


\node [coordinate, below = 5 em of P] (C) {};


\node[coordinate, right = 10 em of C.east](br){};
\node[coordinate, output, left = 10 em of C.west](bl){};
\draw (P.east -| br) node[coordinate] (ur){};
\draw (P.west -| bl) node[coordinate] (ul){};

\draw [-] (P.east) --node[above]{State $\xiv(t)$} (ur);
\draw [-] (br) --node[below]{Sample $\xiv(t_i)$} (bl);
\draw [->] (ul) --node[above]{Command $\Km\hat{\xiv}(t)$} (P.west);

\path (ur) -- coordinate[midway](mr) (br);
\path (ul) -- node[rectangle, midway, anchor=center, draw](ml){Controller} (bl);
\draw [->] (C) -| (ml);
\draw [->] (ml) |- (P);

\node[coordinate, above = 0.75 em of mr] (rbegs){};
\node[coordinate, below = 1.5 em of rbegs](rends){};
\node[coordinate, below left = 1.5 em of rbegs](rups){};
\draw [-] (ur) -- (rbegs);
\draw [-] (rbegs) -- (rups);
\draw [-] (rends) -- (br);
\node[coordinate, left = 0.75 em of rbegs](rarrs){};
\node[coordinate, below = 0.00 em of mr] (rarre){};
\draw [-{Latex[length=0.35em,width=0.25em]}] (rarrs) to [out=-90, in=180, looseness=1] (rarre);

\node[coordinate, above = 0.00 em of ml] (larre){};

\node[coordinate, left = 1. em of mr] (sr) {};
\node[coordinate, right = 1. em of ml] (sl) {};
\node[circle, inner sep = 0, minimum size=2pt, draw, fill=black, below = 0.75 em of ur] (wiretop) {};
\node[left = 4 em of mr, anchor=center, text centered, draw, solid, text width=4em](T){Condition checker} (C);
\draw[->] (wiretop) -| (T);
\draw[dash pattern={on 0.2em off 0.2em}] (T) -- (rarrs);
%
\end{tikzpicture}
		\caption{\label{fig:block} Block diagram of an ETC system.}
	\end{center}
\end{figure}
Figure \ref{fig:block} depicts a simple diagram of a system with ETC. We consider the family of \emph{quadratic triggering conditions} from \cite{heemels2013periodic} with an additional maximum inter-sample time condition:%
\begin{equation}\label{eq:quadtrig}
	t_{i+1} = \inf\left\{\!k>t_i, k \in \N \ \middle| \
		\begin{bmatrix}\xiv(k) \\ \xiv(t_i)\end{bmatrix}\tran
		\!\Qm \begin{bmatrix}\xiv(k) \\ \xiv(t_i)\end{bmatrix} > 0 \
		\text{ or } \ k-t_i \geq \bar{k}
	\right\}\!,
\end{equation}
where $\Qm \in \S^{2\nx}$ is the designed triggering matrix, and $\bar{k}$ is the chosen maximum inter-sample time.\footnote{Typically, a maximum inter-sample time exists naturally for a system with (P)ETC (see \cite{gleizer2018selftriggered}). Still, one may want to set a smaller maximum inter-sample time so as to establish a ``heart beat'' of the system. In any case, this is a necessity if one wants to obtain a finite-state simulation of the system, which is what we do in this paper.} 
Observing Eq.\ \eqref{eq:quadtrig}, we note that the inter-sample time $t_{i+1} - t_{i}$ is a function of $\xv_i \coloneqq \xiv(t_i)$; denoting $\kappa \coloneqq (t_{i+1}-t_i)$ as the inter-sample time, it follows that
\begin{gather}
	\kappa(\xv_i) = \min\left\{k \in \{1, 2, ...\bar{k}\} \mid \xv_i\tran\Nm(k)\xv_i > 0 \text{ or } k=\bar{k}\right\}, \nonumber\\
	\Nm(k) \coloneqq \begin{bmatrix}\Mm(k) \\ \I\end{bmatrix}\tran
	\Qm \begin{bmatrix}\Mm(k) \\ \I\end{bmatrix}, \label{eq:petc_time}\\
	\!\!\Mm(k) \coloneqq \Am_\d(k) + \Bm_\d(k)\Km \coloneqq \e^{\Am hk} + \int_0^{hk}\e^{\Am\tau}\d\tau \Bm\Km.\!\!\nonumber
\end{gather}
where $\I$ denotes the identity matrix. Thus, the event-driven evolution of sampled states can be compactly described by the recurrence
\begin{equation}\label{eq:samples}
	\xiv(t_{i+1}) = \Mm(\kappa(\xiv(t_i))\xiv(t_i).
\end{equation}
Clearly, each initial condition $\xv_0 \in \R^\nx$ leads to infinite sequences of samples $\{\xv_i\}$ and inter-sample times $\{k_i(\xv_0)\}$, defined recursively as 
\begin{equation}\label{eq:samplesytem}
	\begin{aligned}
		\xv_{i+1} &= \Mm(\kappa(\xv_i))\xv_i \\
		k_i(\xv_0) &\coloneqq \kappa(\xv_i).
	\end{aligned}
\end{equation}

Therefore, we can attribute an \emph{average inter-sample time} (AIST) to every initial state:
$$ \text{AIST}(\xv) \coloneqq \liminf_{n\to\infty}\frac{1}{n+1}\sum_{i=0}^{n}k_i(\xv). $$
Using $\liminf$ instead of $\lim$ lets us use the limit lower bound in case the regular limit does not exist, making the AIST metric well-defined.

\fakeparagraph{Objective of this paper} We want to devise a method to compute the exact \emph{smallest average inter-sample time} (SAIST) of the PETC system \eqref{eq:plant}--\eqref{eq:quadtrig}; i.e., the minimal AIST across all possible initial conditions:
\begin{equation}\label{eq:saist}
	\text{SAIST} \coloneqq \inf_{\xv \in \R^\nx}\liminf_{n\to\infty}\frac{1}{n+1}\sum_{i=0}^{n}hk_i(\xv).
\end{equation}
Furthermore, we want to understand the cases where the exact SAIST computation is not possible, and quantify the estimation error if the best we can obtain is an approximation.

The way we define SAIST implies that we do not expect that a system's AIST is irrespective of its initial conditions; as we shall see later in §\ref{sec:numerical}, it is possible that multiple AISTs are observed. Hence, in these cases, we conservatively take the smallest possible one. We argue that the SAIST is an adequate --- in fact, fundamental --- metric to inform designers about the average resource utilization that an ETC implementation is expected to achieve. However, the mere application of Eq.\ \eqref{eq:saist} is largely unpromising: how can one choose a sufficiently large $n$, %
or how can one exhaustively search for initial states to obtain one that yields the SAIST? For this reason, we approach the SAIST computation problem through finite-state abstractions, which we introduce next.

\section{BACKGROUND AND PRELIMINARY RESULTS}\label{sec:prem}

An \emph{abstraction} is a simpler description of a system that preserves desired properties. When working with abstractions, we refer to the original system as the \emph{concrete} system. In this paper, we work with \emph{finite-state abstractions} using the framework of \cite{tabuada2009verification} and its \emph{transition systems}. Later, we equip these systems with weights following \cite{chatterjee2010quantitative}, which allows us to derive metrics such as the SAIST. We then present a special type of finite-state abstraction that preserves SAIST, which we introduced in \cite{gleizer2021hscc}.

\subsection{Transition systems and abstractions}\label{ssec:ts}

In \cite{tabuada2009verification}, Tabuada presents the notion of generalized transition systems:
\begin{defn}[Transition System \cite{tabuada2009verification}]\label{def:system} 
	A system $\Ss$ is a tuple $(\Xs,\Xs_0,\Es,\Ys,H)$ where:
	\begin{itemize}
		\item $\Xs$ is the set of states,
		\item $\Xs_0 \subseteq \Xs$ is the set of initial states,
		\item $\Es \subseteq \Xs \times \Xs$ is the set of edges, or transitions,
		\item $\Ys$ is the set of outputs, and
		\item $H: \Xs \to \Ys$ is the output map.
	\end{itemize}
\end{defn}
Here we have omitted the action set $\Us$ from the original definition because we are solely interested in autonomous systems like \eqref{eq:samplesytem}. %
A system is said to be \emph{finite-state} (\emph{infinite-state}) if the cardinality of $\Xs$ is finite (infinite).  
System $\Ss$ is said to be \emph{non-blocking} if $\forall x \in \Xs, \exists x' \in \Xs : (x,x') \in \Es$. %
We call $x_0x_1x_2...$ an \emph{infinite internal behavior}, or \emph{run} of $\Ss$ if $x_0 \in \Xs_0$ and $(x_i,x_{i+1}) \in \Es$ for all $i \in \N$, and $y_0y_1...$ its corresponding \emph{infinite external behavior}, or \emph{trace}, if $H(x_i) = y_i$ for all $i \in \N$. We denote by $B_{\Ss}(r)$ the external behavior from a run $r = x_0x_1...$ (in the case above, $B_{\Ss}(r) = y_0y_1...$), by $\Bs^\omega_x(\Ss)$ the set of all infinite external behaviors of $\Ss$ starting from state $x$, and by $\Bs^\omega(\Ss) \coloneqq \bigcup_{x\in\Xs_0}\Bs^\omega_x(\Ss)$ the set of all infinite external behaviors of $\Ss$. Finally, $\Bs^{\leq n}(\Ss)$ is the set of all prefixes of length $\leq n$ of each trace in $\Bs^\omega(\Ss)$ (equivalently, the set of its $(\leq n)$-long external behaviors), and $\Bs^+(\Ss)$ is the set of all finite prefixes in $\Bs^\omega(\Ss)$.
A finite sequence $\beta$ is called \emph{transient} if there exists a finite $l$ such that $\gamma\beta\alpha \in \Bs^\omega(\Ss)$ implies that $|\gamma| \leq l$ and $\beta$ is not a subsequence of $\alpha$; equivalently, $\beta$ cannot occur infinitely often in any infinite behavior of $\Ss$.

The ideas of simulation and bisimulation are paramount to establish formal relations between two transition systems.

\begin{defn}[Simulation Relation \cite{tabuada2009verification}]\label{def:sim}
	Consider two transition systems $\Ss_a$ and $\Ss_b$ with $\Ys_a$ = $\Ys_b$. A relation $\Rs \subseteq \Xs_a \times \Xs_b$ is a simulation relation from $\Ss_a$ to $\Ss_b$ if the following conditions are satisfied:
	\begin{enumerate}
		\item[i)] for every $x_{a0} \in \Xs_{a0}$, there exists $x_{b0} \in \Xs_{b0}$ with $(x_{a0}, x_{b0}) \in \Rs;$
		\item[ii)] for every $(x_a, x_b) \in \Rs, H_a(x_a) = H_b(x_b);$
		\item[iii)] for every $(x_a, x_b) \in \Rs,$ we have that $(x_a, x_a') \in \Es_a$ implies the existence of $(x_b, x_b') \in \Es_b$ satisfying $(x_a', x_b') \in \Rs.$
	\end{enumerate}
\end{defn}
When there exists a simulation relation from $\Ss_a$ to $\Ss_b$, we say that $\Ss_b$ simulates $\Ss_a$, denoted by $\Ss_a \preceq \Ss_b$. When $\Rs$ is a simulation relation from $\Ss_a$ to $\Ss_b$ and $\Rs^{-1}$ is a simulation relation from $\Ss_b$ to $\Ss_a$, we say that $\Ss_a$ and $\Ss_b$ are \emph{bisimilar}, denoted by $\Ss_a \bisim \Ss_b$. %
Weaker but important relations associated with simulation and bisimulation are, respectively, \emph{behavioral inclusion} and \emph{behavioral equivalence}:
\begin{defn}[Behavioral inclusion and equivalence \cite{tabuada2009verification}]
	Consider two systems $\Ss_a$ and $\Ss_b$ with $\Ys_a$ = $\Ys_b$. We say that $\Ss_a$ is \emph{behaviorally included} in $\Ss_b$, denoted by $\Ss_a \preceq_\Bs \Ss_b$, if $\Bs^\omega(\Ss_a) \subseteq \Bs^\omega(\Ss_b).$ In case $\Bs^\omega(\Ss_a) = \Bs^\omega(\Ss_b),$ we say that $\Ss_a$ and $\Ss_b$ are \emph{behaviorally equivalent}, which is denoted by $\Ss_a \bisim_\Bs \Ss_b$.
\end{defn}
(Bi)simulations imply behavioral inclusion (equivalence):
\begin{thm}[\cite{tabuada2009verification}]
	Given two systems $\Ss_a$ and $\Ss_b$ with $\Ys_a$ = $\Ys_b$:
	\begin{itemize}
		\item $\Ss_a \preceq \Ss_b \implies \Ss_a \preceq_\Bs \Ss_b$;
		\item $\Ss_a \bisim \Ss_b \implies \Ss_a \bisim_\Bs \Ss_b$.
	\end{itemize}
\end{thm}

The main difference between simulation and behavioral inclusion is that, in the former, a relationship between states must be established: every transition in the concrete system must have at least one matching transition in the abstraction leading to related states. Behavioral inclusion is oblivious to state-based descriptions of a system: all one needs is that all traces observed in the concrete system can also be observed in the abstraction. %
A way of building an abstraction based on behavioral inclusion is through an $l$-complete model:
\begin{defn}[(Strongest) $l$-complete abstraction (adapted from \cite{moor1999supervisory, schmuck2015comparing})] \label{def:slca}
	Let $\Ss \coloneqq (\Xs, \Xs_0, \Es, \Ys, H)$ be a transition system, and let $\Xs_l \subseteq \Ys^l$ be the set of all $l$-long subsequences of all behaviors in $\Bs^\omega(\Ss)$. Then, the system $\Ss_l = (\Xs_l, \Bs^l(\Ss), \Es_l, \Ys, H)$ is called the \emph{(strongest) $l$-complete abstraction (S$l$CA)} of $\Ss$, where
	\begin{itemize}
		\item $\Es_l = \{(k\sigma, \sigma k') \mid k,k' \in \Ys, \sigma \in \Ys^{l-1}, k\sigma, \sigma k' \in \Xs_l\}.$
		\item $H(k\sigma) = k$.
	\end{itemize}
\end{defn}
The idea behind the S$l$CA is to encode the states as the $l$-long behavior fragments of the concrete system. The transitions follow the ``\emph{domino rule}'': e.g., if the last 4 elements of the behavior up to a given time are $abcd$, after one step the first 3 elements must be $bcd$; thus, from having observed $abcd$ alone, a transition from state $abcd$ can lead to any state starting with $bcd$. Finally, the output of a state is its first element. An example of successive $l$-complete approximations of a system with $\Ys = \{1,2\}$ is given in Fig.~\ref{fig:Sl}.

\begin{figure}[tb]
	\begin{center}
		\begin{footnotesize}
			\begin{tikzpicture}[->,>=stealth',shorten >=1pt, auto, node distance=1.5cm,
				semithick]
				\tikzset{every state/.style={minimum size=3em, inner sep=2pt}}
				
				\node[state] 		 (1)    {1};
				\node[state]         (2) [below of=1] {2};
				
				\path (1) edge[bend right] (2) (2) edge[bend right] (1)
				(1) edge [loop left] (1)
				(2) edge [loop left] (2);
				
				\node[state] 		 (12) [right of=1]  {1,2};
				\node[state]		 (21) [below of=12] {2,1};
				\node[state]		 (22) at ([shift=({-30:1.7 cm})]12) {2,2};
				
				\path (12) edge (22) (22) edge [loop above] (22) (22) edge (21) (21) edge[bend right] (12) (12) edge[bend right] (21);
				
				\node[state] 		 (122) [right=2.3cm of 12]  {1,2,2};
				\node[state] 		 (212) [right of=122]     {2,1,2};
				\node[state] 		 (222) [below of=122]     {2,2,2};
				\node[state] 		 (221) [right of=222]     {2,2,1};
				
				\path (221) edge (212) (212) edge (122) (122) edge (222) edge (221) (222) edge[loop left] (222) (222) edge (221);

			\end{tikzpicture}
		\end{footnotesize}
		\caption{\label{fig:Sl} Example of $l$-complete PETC traffic models, for $l=1$ (left), $l=2$ (middle), and $l=3$ (right).}
		\vspace{-1.5em}
	\end{center}
\end{figure}
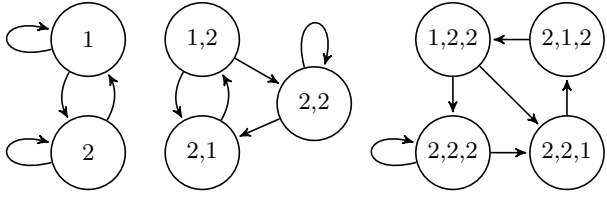

\begin{rem} We have made an adaptation from the original definition from \cite{moor1999supervisory}, where the system is defined on a behavioral framework \cite{willems1991paradigms}; here we present directly a realization of the S$l$CA as a transition system according to Def.~\ref{def:system}. Schmuck et al \cite{schmuck2015comparing} showed that different realizations exist for the S$l$CA	of a system, depending on whether you encode states based on past, future, or a mix of past and future observations. In Def.~\ref{def:slca}, we pick the one based on future observations, which simplifies the encoding (all states are $l$-long sequences without the need for "no-output yet" characters, see \cite{moor1999supervisory}), and is the tightest abstraction from a simulation relation perspective (see \cite[Thm.~5]{schmuck2015comparing}).
\end{rem}

In \cite[Theorem 9]{schmuck2015comparing}, it is concluded that a quotient-based approach \cite{tabuada2009verification} can create an abstraction bisimilar to the S$l$CA in case the concrete system is future-unique, which is the case of deterministic systems. Thus, we shall use the term $l$-complete for quotient-based abstractions whose states represent the next $l$ outputs of their related concrete states. How to do it will become clear in §\ref{ssec:lcomplete}, where we build the abstractions of the PETC traffic model. With this in mind, the following fact is a direct consequence of Theorems 6 and 7 from \cite{schmuck2015comparing}.
\begin{prop}\label{prop:schmuck}
	Consider a deterministic system $\Ss$ and its S$l$CA $\Ss_l$ from Definition \ref{def:slca}, for some $l \geq 1$. Then, $\Ss \preceq \Ss_{l+1} \preceq \Ss_l.$
\end{prop}
Prop.~\ref{prop:schmuck} gives that $l$-complete abstractions provide a framework of obtaining simulations and their refinements. It is not a surprising result, since encoding states with more elements of the concrete system's behavior constrains the set of behaviors it can generate, even if it increases the number of states in the abstraction.
\begin{rem}\label{rem:lbisim} Bisimulation is obtained when $\Ss_{l+1} = \Ss_l$ (modulo the names of the states); it is trivial to see that this only happens when abstracting an autonomous deterministic system if the abstraction is deterministic. In addition, $\lim_{l\to\infty}\Ss_l \bisim_\Bs \Ss$.
\end{rem}

\subsection{Quantitative automata}

While much of the field of formal methods in control is concerned with qualitative analysis, such as establishing safety, stability, and reachability, often quantitative computations are of interest: examples are computing the decay rate, the maximum overshoot, or our case, the average sampling period of an ETC system.  In \cite{chatterjee2010quantitative}, Chatterjee et al.~established a comprehensive framework for quantitative problems on finite-state systems, from which we borrow some definitions and results, while adjusting notation to keep consistency with the previous section.

\begin{defn}[Weighted transition system (adapted from \cite{chatterjee2010quantitative})]
	A \emph{weighted transition system} (WTS) $\Ss$ is the tuple $(\Xs,\Xs_0,\Es,\Ys,H, \gamma)$, where
	\begin{itemize}
		\item $(\Xs,\Xs_0,\Es,\Ys,H)$ is a \emph{non-blocking} transition system;
		\item $\gamma: \Es \to \Q$ is the \emph{weight function}.
	\end{itemize}
\end{defn}
The notation adjustment we have made is including outputs to comply with Tabuada's transition systems; again, we ignore the action set as our scope is limited to autonomous systems. 

Given a run $r = x_0x_1...$ of $\Ss$, we abuse notation denoting by $\gamma(r) = v_0v_1...$ the \emph{sequence of weights} defined by $v_i = \gamma(x_i,x_{i+1})$. A \emph{value function} $\text{Val} : \Q^\omega \to \R$ attributes a value to an infinite sequence of weights $v_0v_1...$. Among the well-studied value functions, the one of our interest is
$$ \LimAvg(v) \coloneqq \liminf_{n\to\infty}\frac{1}{n+1}\sum_{i=0}^{n}v_i.$$
Similarly, for a finite sequence $v$ of length $n$, let $\Avg(v) \coloneqq \frac{1}{n+1}\sum_{i=0}^{n}v_i$. We define the smallest and largest LimAvg values of an automaton respectively as $\VS(\Ss) \coloneqq \inf\{\LimAvg(\gamma(r)) \mid r \text{ is a run of } \Ss\}$ and $\VL(\Ss) \coloneqq \sup\{\LimAvg(\gamma(r)) \mid r \text{ is a run of } \Ss\}.$ Clearly $\VS(\Ss) = -\VL(-\Ss),$ where we denote by $-\Ss$ the WTS $\Ss$ with all of its weights negated; thus, we focus on the results for $\VS$ in what follows. The following theorem is essentially an excerpt from Theorem 3 in \cite{chatterjee2010quantitative}, which uses the classical result from Karp \cite{karp1978characterization}:
\begin{thm}\label{thm:limavg}
	Given a finite-state WTS $\Ss, \VS(\Ss)$ can be computed in $\bigO(|\Xs||\Es|)$. Moreover, system $\Ss$ admits a cycle $x_0x_1...x_k$ satisfying $x_i \to x_{i+1}, i < k,$ and $x_k \to x_0$ s.t. $\LimAvg(\gamma((x_0x_1...x_k)^\omega)) = \VS(\Ss)$. 
\end{thm}
The cycle mentioned above is a \emph{smallest-in-average cycle} (SAC) of the weighted digraph defined by $\Ss$, and can be recovered in $\bigO(|\Xs|)$ using the procedure of \cite{chaturvedi2017note}.

\subsection{Quantitative verification through abstractions}

In \cite{gleizer2021hscc}, we have presented some basic results about the relationship between the SAISTs of a system and its abstraction. First, we start with a simplifying condition for weighted transition systems: a WTS is called \emph{simple} if for all $(x,x') \in \Es, \gamma(x,x') = H(x)$, i,e., the weight of a transition is equal to the output of its outbound state. Throughout this paper, when working with a transition system with $\Ys \subset \Q$, we omit the weight function $\gamma$, implying that we have a simple WTS. Here, we recall results from \cite{gleizer2021hscc}.

\begin{prop}[\cite{gleizer2021hscc}]\label{prop:bound}
	If two simple WTSs $\Ss_a$ and $\Ss_b$ satisfy $\Ss_a \mathrel{\preceq_\Bs\!(\bisim_\Bs)} \Ss_b$, then $\VS(\Ss_a) \mathrel{\geq\!(=)} \VS(\Ss_b)$ and $\VL(\Ss_a) \mathrel{\leq\!(=)} \VL(\Ss_b)$.
\end{prop}

Proposition \ref{prop:bound} uses the fact that the WTSs are simple to relate weight runs and behaviors. With it, abstractions that simulate the concrete system provide a way to underestimate the SAIST and overestimate the LAIST, thanks to Theorem \ref{thm:limavg}. Equality can be achieved with the following type of abstraction.
\begin{defn}[Smallest-average-cycle-equivalent simulation \cite{gleizer2021hscc}] \label{def:macesim} Consider two simple WTSs $\Ss_a$ and $\Ss_b$ satisfying $\Ss_a \preceq \Ss_b$. Let $\SAC(\Ss_b)$ be the set of smallest-in-average cycles of $\Ss_b$. If there exists a behavior of the form $dc^\omega \in \Bs^\omega(\Ss_a)$ where $d$ is finite and $c \in \SAC(\Ss_b)$, then $\Ss_b$ is a \emph{smallest-average-cycle-equivalent (SACE) simulation} of $\Ss_a$.
\end{defn}
A SACE simulation is a normal simulation with the added requirement that at least one of the SACs of the abstraction is an actual recurrent behavior of the concrete system, after some finite transient. Clearly, SACE simulation is stronger than simulation 
but significantly weaker than bisimulation. Equivalently, a \emph{largest-average-cycle-equivalent simulation}, or LACE simulation, can be defined using the maximum average cycle instead. The following result is a straightforward conclusion from Proposition \ref{prop:bound} and Theorem \ref{thm:limavg}.
\begin{prop}[\cite{gleizer2021hscc}]\label{prop:sace}
	Consider two simple WTSs $\Ss_a$ and $\Ss_b$; if $\Ss_b$ is a finite-state SACE simulation of $\Ss_a$, then $\VS(\Ss_a) = \VS(\Ss_b)$.
\end{prop}

\begin{rem}\label{rem:simplenotneeded}
	In fact, to use Def.\ \ref{def:macesim} and Prop.\ \ref{prop:sace}, it is not needed that the WTSs are simple. One can always turn a WTS into an equivalent simple one by adding artificial states: suppose that $(x,y)$ and $(x,z)$ belong to $\Es$ and $\gamma(x,y) = a \neq \gamma(x,z) = b.$ Add artificial states $y'$ and $z'$ and replace the aforementioned transitions with $(x,y'), (x,z'), (y',y), (z',z)$, setting $\gamma(x,y') = \gamma(x,z') = 0$, $\gamma(y',y) = a$ and $\gamma(z',z) = b$. Applying this procedure to the whole system gives a simple WTS, and again behaviors are equal to sequences of weights. The LimAvg value of any run of this modified system is half of the value of the original equivalent run (since we are adding zeros at every other transition).  
\end{rem}

\begin{rem}\label{rem:sacecontrollable} The concept of SACE simulation and its related results have been recently expanded to non-autonomous WTSs in \cite{gleizer2021cdc}, where the objective was to \emph{design} sampling strategies instead of evaluating them. This expansion is not needed for the scope of this paper.
\end{rem}

For the cases where obtaining a SACE simulation of $\VS(\Ss_a)$ is not possible, one may still be interested in computing an estimate of the error $\VS(\Ss_a) - \VS(\Ss_b).$ In \cite{gleizer2021hscc}, the maximal value $\VL(\Ss_b)$ was used to this end, but a better approximation can be found by inspecting the maximal average cycle of the attractors of $\Ss_b$. 

\begin{prop}\label{prop:upperbound}
	Let $\Ss_a \coloneqq (\Xs_a, \Xs_a, \Es_a, \Ys, H_a)$ and $\Ss_b \coloneqq (\Xs_b, \Xs_b, \Es_b, \Ys, H_b)$ be two simple WTSs, $\Rs$ be a simulation relation from $\Ss_a$ to $\Ss_b$, and $\As \subset \Xs_b$ be a strongly forward invariant set %
	\footnote{A strongly forward invariant set $\As \subseteq \Xs$ is a set that satisfies $\forall x \in \As, (x, x') \in \Es \implies x' \in \As$.}%
	of $\Ss_b$. If there exists $x_b \in \As$ such that $(x_a, x_b) \in \Rs$ for some $x_a \in \Xs_a$, then $\VS(\Ss_a) \leq \VL((\As, \As, \Es_b, \Ys, H_b)) \leq \VL(\Ss_b)$.
\end{prop}

\begin{proof}
	First, it is a simple exercise to see that $(\Xs,\Xs',\Es,\Ys,H) \preceq (\Xs,\Xs,\Es,\Ys,H)$ if $\Xs' \subseteq \Xs$. Now, take $(x_a,x_b) \in \Rs$ where $x_b \in \As$. Then, $(\Xs_a, \{x_a\}, \Es_a, \Ys, H_a) \preceq \Ss_a$. 
	At the same time, with the same relation $\Rs$ we can verify that $(\Xs_a, \{x_a\}, \Es_a, \Ys, H_a) \preceq (\Xs_b, \{x_b\}, \Es_b, \Ys, H_b)$. Therefore, by Prop.~\ref{prop:bound}, $\VS((\Xs_a, \{x_a\}, \Es_a, \Ys, H_a)) \geq \VS(\Ss_a),$ and $\VL((\Xs_b, \{x_b\}, \Es_b, \Ys, H_b)) \geq \VL((\Xs_a, \{x_a\}, \Es_a, \Ys, H_a))$. Because $\VL(\cdot) \geq \VS(\cdot)$, we get that $\VL((\Xs_b, \{x_b\}, \Es_b, \Ys, H_b)) \geq \VS(\Ss_a).$
	
	Now, because $\As$ is strongly forward invariant, every run of $(\Xs_b, \{x_b\}, \Es_b, \Ys, H_b)$ contains only states in $\As$. Thus, $(\Xs_b, \{x_b\}, \Es_b, \Ys, H_b) \bisim_\Bs (\As, \{x_b\}, \Es_b, \Ys, H_b) \preceq (\As, \As, \Es_b, \Ys, H_b)$. Then, applying Prop.~\ref{prop:bound} again gives $\VL((\As, \As, \Es_b, \Ys, H_b)) \geq \VS(\Ss_a).$
	
	Finally, because $(\As, \As, \Es_b, \Ys, H_b) \preceq \Ss_b,$ Prop.~\ref{prop:bound} also gives that $\VL(\Ss_b) \geq \VL((\As, \As, \Es_b, \Ys, H_b))$.
\end{proof}

When the abstraction $\Ss_b$ is finite, its smallest strongly invariant sets are simply the attractive strongly connected components (SCCs) of the graph associated with $\Ss_b$. Obtaining the SCCs of a graph with $n$ vertices and $m$ edges has complexity $\bigO(n+m)$ \cite{karp1978characterization} and in fact is part of the steps to compute its smallest (or largest) average cycle.

\section{LIMIT AVERAGE FROM $l$-COMPLETE ABSTRACTIONS}\label{sec:genericlimavg}

In this section we provide some results on the computation of the infimal limit average of a simple WTS $\Ss$ through the use of its S$l$CA $\Ss_l$. The first result is an obvious conclusion from combining Prop.~\ref{prop:schmuck} with \ref{prop:bound}:
\begin{prop}\label{prop:lgivesbound}
	Consider a simple WTS $\Ss$ and its S$l$CA $\Ss_l$ (Def.~\ref{def:slca}), for some $l \geq 1$. It holds that $\VS(\Ss_l) \leq \VS(\Ss).$
\end{prop}

Considering the idea of SACE simulation, a simple conceptual algorithm that can successfully compute $\VS(\Ss)$ is given in Alg.~\ref{alg:general}. The idea is to increment $l$ until the smallest-in-average cycle of $\Ss_l$ is verified in the concrete system. The algorithm requires one to be able to compute the S$l$CA of a given system (line 3) and to verify the existence of periodic behavior (line 5); these steps will be discussed for PETC traffic on §\ref{sec:mainpetc}. As we will see now, Alg.~\ref{alg:general} is in fact a semi-algorithm; depending on the behavior of $\Ss$, it may not terminate. 
\begin{algorithm}\caption{\label{alg:general}Computation of $\VS(\Ss)$}
	\begin{flushleft}
		\hspace*{\algorithmicindent} \textbf{Input:} A simple WTS $\Ss$ with $\Ys \subset \Q, |\Ys| < \infty$ \\
		\hspace*{\algorithmicindent} \textbf{Output:} $l, \Ss_l, \mathtt{V}, \sigma$
	\end{flushleft}
	\begin{algorithmic}[1]
		\State $l \gets 1$
		\While{true}
		\State Build $\Ss_l$ \Comment{(Def.~\ref{def:slca})}
		\State $\mathtt{V} \gets \VS(\Ss_l), \ \sigma \gets \SAC(\Ss_l)$ \Comment{\cite{karp1978characterization, chaturvedi2017note}}	
		\If{$\sigma^\omega \in \Bs^\omega(\Ss)$}
		\State \Return
		\EndIf
		\State $l \gets l+1$
		\EndWhile
	\end{algorithmic}
\end{algorithm}
The following result shows under which conditions there is a finite $l$ such that $\VS(\Ss_l) = \VS(\Ss).$
\begin{thm}\label{thm:lcompleteisoptimal}
	Consider a simple finite WTS $\Ss$ and assume that there exists a finite $m \in \N$ 
	such that every infinite behavior $\alpha \in \Bs^\omega(\Ss)$ satisfies $\Avg(\beta) \geq \VS(\Ss)$, for every non-transient subsequence $\beta$ of $\alpha$ with $|\beta| = m$.  
	Then there exists a finite $l$ such that 
	the $l$-complete simulation $\Ss_l$ of $\Ss$ satisfies $\VS(\Ss_l) = \VS(\Ss)$.
\end{thm}

\begin{proof}
	First we prove that, if $\beta$ is transient, then there exists $l$ large enough such that $\beta$ cannot be a subsequence of $\sigma^\omega$ for any cycle $\sigma$ of $\Ss_l$. For that, suppose by contradiction that, $\forall L, \exists l \geq L$ for which a cycle $\sigma$ of $\Ss_l$ exists s.t.~$\beta$ \emph{is} a subsequence of $\sigma^\omega$; w.l.o.g., assume that $l > m$. Then, there exists a word $\gamma\beta$ of length $l$ that is a subsequence of $\sigma^\omega$; hence, $|\gamma| = l-m$. This holds because for any natural number $p$, $\sigma^p\beta$ is a subsequence of $\sigma^p\sigma^\omega = \sigma^\omega.$ Now, by definition of $\Ss_l$, $\gamma\beta \in \Bs^l(\Ss)$. Since $l$ can be chosen arbitrarily large, $\beta$ can occur arbitrarily late in a behavior of $\Ss$, thus contradicting the fact that it is transient.
	
	Therefore, there exists $l$ large enough such that, for every cycle $\sigma$ of $\Ss_l$, every $m$-long subsequence $\beta$ of $\sigma^\omega$ is non-transient. From Theorem \ref{thm:limavg}, one such cycle satisfies $\VS(\Ss_l) = \LimAvg(\sigma^\omega).$ Let $p \coloneqq |\sigma|$. Then, $\sigma^m$ has length $pm$ and as such it can be divided in $p$ non-transient subsequences $\beta_i$, not necessarily distinct, of length $m$. Now,
	$$ \VS(\Ss_l) = \LimAvg(\sigma^\omega) = \LimAvg((\sigma^m)^\omega)) = \Avg(\sigma^m) = \frac{1}{p}\sum_{i=1}^p\Avg(\beta_i) \geq \VS(\Ss). $$
	Since, by Prop.~\ref{prop:lgivesbound}, $\VS(\Ss_l) \leq \VS(\Ss)$, it holds that $\VS(\Ss_l) = \VS(\Ss)$.
\end{proof}

Theorem \ref{thm:lcompleteisoptimal} states that it is sufficient for it to exist an $m$ large enough such that every ``persistent'' $m$-long behavior fragment $\beta$ of $\Ss$ has higher or equal average than $\VS(\Ss)$. Intuitively, constraining the assumption of $\beta$ occurring infinitely often has the idea of excluding transient behaviors $\beta$, which do not affect the LimAvg value. For cases where $\beta$ can occur infinitely often in some behavior, but $\beta^\omega$ is not a behavior of $\Ss$, one can construct counterexamples in which $\VS(\Ss_l) < \VS(\Ss)$ for all $l$:
\begin{ex}\label{ex:badword}
	Consider a system $\Ss$ with behavior set $\Bs^\omega(\Ss) = \{(1^n2^n)^\omega \mid n \in \N\}$. Obviously, $\VS(\Ss) = 1.5$. However, for any $l$, $(1^l)^\omega \in \Bs^\omega(\Ss_l)$, hence $\VS(\Ss_l) = 1$ for any $l$.
\end{ex}
\begin{ex}\label{ex:irrational}
	Consider the system $\Ss = ([0,1], [0,1], \Es, \{0, 1\}, H)$ where $\Es = \{(x, x + a \bmod 1)\}$ and $H(x) = 1$ if $x < a$ and 0 otherwise. When $a$ is irrational, $\Ss$ is called an irrational rotation. Because it is ergodic with respect to the Lebesgue measure \cite{demelo2012one}, $\LimAvg(\alpha) = a$ for any $\alpha \in \Bs^\omega(\Ss)$. Thus, $\VS(\Ss) = a$ is irrational. Since for every finite $l$, $\VS(\Ss_l)$ is a rational number (as a consequence of Theorem \ref{thm:limavg} and the fact that $\Ss_l$ is finite), $\VS(\Ss_l) \neq \VS(\Ss).$ Finally, from Prop.~\ref{prop:lgivesbound}, $\VS(\Ss_l) \leq \VS(\Ss),$ thus $\VS(\Ss_l) < \VS(\Ss)$ for all finite $l$.
\end{ex}

Note that, for Ex.~\ref{ex:irrational}, the minimum number of 1s in a behavior fragment of length $n$ is $\floor*{na}$, hence $\VS(\Ss_l) = \frac{\floor*{la}}{l}$, which asymptotically approaches $a$ as $l$ goes to infinity. For Ex.~\ref{ex:badword}, we cannot obtain this asymptotic approximation.

The conditions in Theorem \ref{thm:lcompleteisoptimal} do not imply that the SAC $\sigma$ of $\Ss_l$ satisfies $\sigma^\omega \in \Bs^\omega(\Ss)$; thus, we may have equality of LimAvg values without a SACE simulation. Therefore, under these conditions, Alg.~\ref{alg:general} can be interrupted with the exact value, but with no certificate that this is the case. Its termination is guaranteed when there is a cyclic minimizing behavior, and additionally that the other behaviors have limit average values strictly larger than that of the cycle:

\begin{thm}\label{thm:Vcomputable}
	Consider a simple WTS $\Ss$, and suppose $\Ss$ satisfies the premises of Theorem \ref{thm:lcompleteisoptimal}. Furthermore, assume there exists an $m$-long sequence $\sigma$  such that $\sigma^\omega \in \Bs^\omega(\Ss),$ and that every non-transient subsequence $\beta,$ $\norm[\beta] = m$ of every behavior $\alpha \in \Bs^\omega(\Ss)$ satisfies $\LimAvg(\beta^\omega) > \VS(\Ss)$ if $\beta$ is not a subsequence of $\sigma^\omega$. Then Alg.~\ref{alg:general} terminates with $\mathrm{V} = \VS(\Ss)$.
\end{thm}
The proof requires some technical results on cyclic permutations of sequences and we leave it for the appendix. The main insight is that the conditions of Theorem \ref{thm:Vcomputable} imply that, for sufficiently large $l$, $\Ss_l$ has only one one SAC $\sigma$, modulo cyclic permutations, which attains the minimum value; at the same time, for large enough $l$, this $\sigma$ satisfies $\sigma^\omega \in \Bs^\omega(\Ss)$. Hereafter, we say that a system satisfying the premises of Theorem \ref{thm:Vcomputable} has an \emph{isolated SAC}. This does not mean that the behavior of $\Ss$ is simple, or that a finite-state bisimulation of it exists:
\begin{ex}\label{ex:doublingmap}
	Consider the \emph{doubling map} system $\Ss = ([0,1], [0,1], \Es, \{0, 1\}, H)$ where $\Es = \{(x, 2x \bmod 1) \mid x \in [0,1]\}$ and $H(x) = 0$ if $x < 1/2$ and 1 otherwise. The behavior of this system is $(0^+1^+)^\omega$, its smallest cycle is $0^\omega$ with value zero (obtained with $x_0 = 0$). This system does not admit a finite-state bisimulation, but its 1-complete abstraction is $\Ss_1 = \{\{0,1\}, \{0,1\}, \{(0,0), (0,1), (1,0), (1,1)\}, \{0, 1\}, \Id\},$ where $\Id$ is the identity operator. Clearly, $\Ss_1$ is a SACE simulation of $\Ss$ (in fact, it is behaviorally equivalent, but not bisimilar). The system $\Ss$ satisfies the premises of Theorem \ref{thm:Vcomputable} with $m=1$.
\end{ex}

Now that we have the general framework for the computation of $\VS(\Ss)$, we see how to apply it for PETC traffic.

\section{COMPUTING THE SAIST OF PETC}\label{sec:mainpetc}

We start by describing the evolution of sampled states and ISTs of a PETC system, cf.~Eq.\ \eqref{eq:samplesytem}, as a transition system following Def.~\ref{def:system}:
\begin{equation}\label{eq:S}
	\begin{aligned}
		\Ss \coloneqq (\R^n&, \R^n, \Es, \Ys, H), \text{ where} \\
		\Es & = \{(\xv,\xv') \in \R^n \times \R^n \mid \xv' = \Mm(\kappa(\xv))\xv\}, \\
		\Ys &= \{1,2,...,\bar{k}\}, \\ 
		H &= \kappa.
	\end{aligned}
\end{equation}

System $\Ss$ is our concrete infinite-state system, for which we want develop an algorithm like Alg.~\ref{alg:general}. For this we need to be able to (i) build an $l$-complete abstraction of the system, (ii) compute its SAC, and (iii) check if its minimum mean cycle exists in the concrete system. 
Naturally, Karp's algorithm \cite{karp1978characterization, chaturvedi2017note} constitute the tool for task (ii). In the next section we present how to obtain $l$-complete abstractions of $\Ss$. Then, in §\ref{ssec:verifycycle}, we show how can a cyclic behavior be verified to be trace of $\Ss$. Finally, we present the full algorithm and discuss its robustness and applicability in subsequent subsections.

\subsection{$l$-complete PETC traffic models}\label{ssec:lcomplete}

As mentioned in §\ref{ssec:ts}, for autonomous deterministic systems such as $\Ss$ from Eq.\ \ref{eq:S}, a quotient-based approach can be used to obtain its S$l$CA $\Ss_l$. %
The idea is to divide the state-space $\Xs$ into regions $\Xs_{y_1y_2...y_l}$, where the first $l$ elements of any behavior in $\Bs_x^\omega(\Ss)$, for any $x \in \Xs_{y_1y_2...y_l}$, are exactly $y_1, y_2, ..., y_l$. If $\Ss$ is deterministic, this division generates a partition, as from one state $x$ there exists only one infinite behavior. In \cite{gleizer2020towards, gleizer2021hscc}, we have used this idea to construct finite-state PETC traffic models abstracting system \eqref{eq:S}, coming up with the following relation:
\begin{defn}[Inter-sample sequence relation \cite{gleizer2021hscc}]\label{def:bisimrel} Given a sequence length $l$, we denote by $\Rs_l \subseteq \Xs \times \Ys^l$ the relation satisfying 
	$(\xv,\dummy{k_1k_2...k_l}) \in \Rs_l$ if and only if
	\begin{subequations}\label{eq:sequence}
		\begin{align}
			\xv &\in \Qs_{k_1}, \label{eq:sequence_k1}\\
			\Mm(k_1)\xv &\in \Qs_{k_2}, \label{eq:sequence_k2}\\
			\Mm(k_2)\Mm(k_1)\xv &\in \Qs_{k_3}, \label{eq:sequence_k3}\\
			& \vdots \nonumber\\
			\Mm(k_{l-1})...\Mm(k_1)\xv &\in \Qs_{k_l}, \label{eq:sequence_kl-1}
		\end{align}
	\end{subequations}
	where 
	\begin{equation}\label{eq:setq}
		\begin{gathered}
			\Qs_k \coloneqq \Ks_k \setminus \left(\bigcap_{j=\underline{k}}^{k-1} \Ks_{j}\right) = \Ks_k \cap \bigcap_{j=1}^{k-1} \bar{\Ks}_{j}, \\
			\Ks_k \coloneqq \begin{cases}
				\{\xv \in \Xs| \xv\tran\Nm(k)\xv > 0\}, & k < \bar{k}, \\
				\R^\nx, & k = \bar{k}.
			\end{cases}
		\end{gathered}
	\end{equation}
\end{defn}

Eq.\ \eqref{eq:setq}, from \cite{gleizer2020scalable}, defines the sets $\Qs_k$, containing the states from which the next trigger happens exactly after $k$ time units. Eq.\ \eqref{eq:sequence} states that a state $\xv \in \R^n$ is related to a state $k_1k_2...k_l$ of the abstraction if its generated inter-sample time sequence for the next $l$ samples is $k_1,k_2,...,k_l$.

\begin{rem} Setting $l=1$ gives a \emph{quotient state set} \cite{tabuada2009verification} of $\Ss$ in \eqref{eq:S}, while larger values of $l$ can be seen as refinements using the bisimulation algorithm of \cite[Chapter 8]{tabuada2009verification}.
\end{rem}
\begin{defn}\label{def:lsim} Given an integer $l \geq 1$, the \emph{$l$-complete PETC traffic model} is the system $\Ss_l \coloneqq \left(\Xs_l, \Xs_l, \Es_l, \Ys, H_l\right)$, with 
	\begin{itemize}
		\item $\Xs_l \coloneqq \pi_{\Rs_l}(\Xs)$,
		\item $\Es_l = \{(k\sigma, \sigma k') \mid k,k' \in \Ys, \sigma \in \Ys^{l-1}, k\sigma, \sigma k' \in \Xs_l\},$
		\item $H_l(k_1k_2...k_l) = k_1.$
	\end{itemize}
\end{defn}
The model above partitions the state-space $\R^\nx$ of the PETC into subsets associated with the next $l$ inter-sample times these states generate, i.e., it is an $l$-complete abstraction, but also a quotient-based model. Computing the state set, $\pi_{\Rs_l}(\Xs)$, requires determining whether or not, for each $k_1k_2...k_l \in \Ys^l$, its associated conjunction of quadratic inequalities in Eq.\ \eqref{eq:sequence} admits a solution $\xv \in \R^{\nx}$; only if it does, then $\sigma \in \Xs_l$. This can be determined using a nonlinear satisfiability-modulo-theories (SMT) solver such as Z3 \cite{demoura2008z3}: for that, the variable is $\xv \in \R^\nx$ and the query is $\exists \xv \in \R^{\nx} : \ $ Eq.\ \eqref{eq:sequence} holds.%
\footnote{Alternatively, this query may be solved approximately through convex relaxations as proposed in \cite{gleizer2020scalable}. Using relaxations implies finding inter-sample sequences that may not be exhibited by the real system. This still generates a simulation relation, but containing more spurious behaviors.} %
The output map $H_l$ is the next sample alone, and the transition relation is based on the domino rule, as in Def.~\ref{def:slca}. %

\subsection{Verifying SACE equivalence}\label{ssec:verifycycle}

In this subsection, we are interested in determining whether a sequence of outputs $(k_1k_2...k_m)^\omega \eqqcolon \sigma^\omega$ is a possible behavior of system $\Ss$ in Eq.\ \eqref{eq:S}. This is equivalent to finding a run $\{\xv_i\}$ whose trace is $\sigma^\omega$. From now on, we denote by $\Qs_\sigma$, or \emph{$\sigma$-cone}, the set of all points $\xv \in \R^\nx$ satisfying Eq.\ \eqref{eq:sequence} with $l=m$ and by $\Mm_\sigma \coloneqq \Mm(k_m)\Mm(k_{m-1})\cdots\Mm(k_1)$. For the formal results, consider the following classes of square matrices:
\begin{defn}[Mixed matrix] Consider a matrix $\Mm \in \R^{n \times n}$ and let $\lambda_i, i \in \N_{\leq n}$ be its eigenvalues sorted such that $\norm[\lambda_i] \geq \norm[\lambda_{i+1}]$ for all $i$. We say that $\Mm$ is \emph{mixed} if, for all $i < n$, $\norm[\lambda_i] = \norm[\lambda_{i+1}]$ implies that $\Im(\lambda_i) \neq 0$ and $\lambda_i = \lambda_{i+1}^*$.
\end{defn}

\begin{rem} Mixed matrices cannot have eigenvalues with the same magnitude, except for complex conjugate pairs. Every mixed matrix is diagonalizable, but the converse does not hold (e.g., the identity is not mixed). The set of mixed matrices is full Lebesgue measure. With a non-pathological choice of $h$, the matrices $\Mm(1), \Mm(2), ... \Mm(\bar{k})$ from Eq.\ \eqref{eq:petc_time} are all mixed, even if $\Km$ is chosen to place poles of $\Am+\Bm\Km$ in the same point of the complex plane; it is sensible (but not guaranteed) to expect that their products are also mixed. From a linear systems perspective, all modes of a mixed matrix have different speeds. \end{rem} 

\begin{defn}[Matrix of irrational rotations] A matrix $\Mm \in \R^{n \times n}$ is said to be of \emph{irrational rotations} if the arguments of all of its complex eigenvalues are irrational multiples of $\piconst$.
\end{defn}

\begin{rem} If $\Mm$ has a pair of complex conjugate eigenvalues whose argument is a rational multiple of $\piconst$, i.e., $p\piconst/q,$ where $p,q \in \N$, then the corresponding eigenvalues of $\Mm^q$ are real. The set of real matrices of \emph{rational} rotations is Lebesgue-measure zero but dense in $\R^{n \times n}$.\end{rem}
	
If $\Mm_\sigma$ is mixed and of irrational rotations, one can verify if $\sigma^\omega$ is a behavior of $\Ss$ from Eq.\ \eqref{eq:S} by checking the linear invariants of $\Mm_\sigma$:

\begin{thm}\label{thm:verifycycle}
	Consider system \eqref{eq:S} and let $\sigma \in \Ys^m, m \in \N,$ be a sequence of outputs. (i) If $\Mm_\sigma$ is nonsingular and there exists a linear invariant $\As$ of $\Mm_\sigma$ such that $\As \setminus \{\O\} \subseteq \Qs_\sigma$, then $\sigma^\omega \in \Bs^\omega(\Ss)$. %
	Moreover, if (ii) $\Mm_\sigma$ is additionally mixed and of irrational rotations, then $\sigma^\omega \in \Bs^\omega(\Ss)$ implies that there exists a linear invariant $\As$ of $\Mm_\sigma$ such that $\As \subseteq \cl(\Qs_\sigma).$ 
\end{thm}

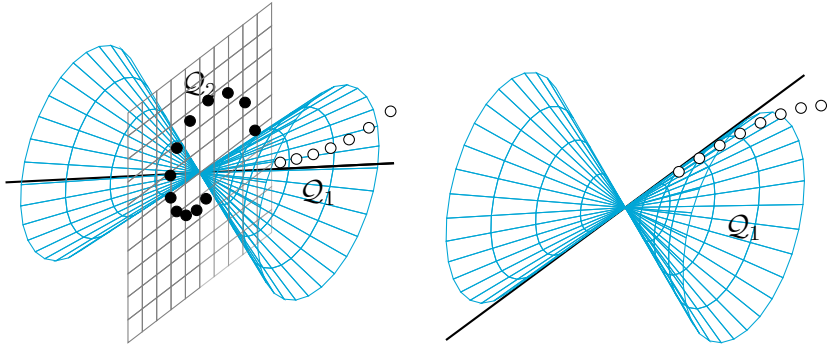
\begin{figure}
	\begin{center}
		\begin{tikzpicture}
\begin{axis}[
	hide axis,
	domain=-1:1,
	y domain=0:-2*pi,
	xmin=-1.5, xmax=1.5,
	ymin=-1.5, ymax=1.5, zmin=-1.2,
	samples=10,
	samples y=30,
	]
	\addplot3[domain=-1.5:1.5, samples y=1, thick]
	({x}, {0.2*x}, {0.2*x});

	\addplot3[mesh, domain=0:1, tudcyan, samples=5]
	({x},{1.1*x*cos(deg(y))},{1.1*x*sin(deg(y))});
	\addplot3[mesh, tudcyan, domain=-1:0, y domain = pi/5:3*pi/2, samples y=20, samples=5,]
	({x},{1.1*x*cos(deg(y + pi))},{1.1*x*sin(deg(y + pi))});
	
	\node[draw=none] at (1,0,0) (q1) {$\Qs_1$};
	\node[draw=none] at (0,0,1) (q2) {$\Qs_2$};
	
	\addplot3[domain=0.6:1.5, samples y=1, thick]
	({x}, {0.2*x}, {0.2*x});
	
	\addplot3[only marks, domain=0.6:1.3, samples y=1, samples=7, mark=*,mark options={fill=white}]
	({x}, {0.2*x + exp(3*(x-1.6))}, {0.2*x + exp(3*(x-1.6))});
	
	\addplot3[mesh, domain=-1.3:1.3, y domain=0:1.3, samples=11, samples y=6, thin, gray]
	({0}, {x}, {y});
	\addplot3[mesh, domain=-1.3:0, y domain=-1.3:0, samples=6, samples y=6, thin, gray]
	({0}, {x}, {y});
	\addplot3[mesh, domain=0:1.3, y domain=-1.3:0, samples=6, samples y=6, thin, gray, opacity=0.5]
	({0}, {x}, {y});
	
	\addplot3[only marks, domain=0:1.6*pi, samples y=1, samples=12, mark=*]
	({0}, {cos(deg(x))*exp(-0.2*(x))}, {sin(deg(x))*exp(-0.2*(x))});

\end{axis}
\end{tikzpicture} \quad %
		\begin{tikzpicture}
	\begin{axis}[
		hide axis,
		domain=-1:1,
		y domain=0:-2*pi,
		xmin=-1.5, xmax=1.5,
		ymin=-1.5, ymax=1.5, zmin=-1.2,
		samples=10,
		samples y=30,
		]
		\addplot3[domain=-1:1, samples y=1, thick]
		({x}, {1.1*x}, {1.1*x});
		
		\addplot3[mesh, domain=0:1, tudcyan, samples=5]
		({x},{1.1*x*cos(deg(y))},{1.1*x*sin(deg(y))});
		\addplot3[mesh, tudcyan, domain=-1:0, y domain = pi/5:3*pi/2, samples y=20, samples=5,]
		({x},{1.1*x*cos(deg(y + pi))},{1.1*x*sin(deg(y + pi))});
		
		\node[draw=none] at (1,0,0) (q1) {$\Qs_1$};

		\addplot3[only marks, domain=0.3:1.1, samples y=1, samples=8, mark=*,mark options={fill=white}]
		({x}, {1.1*x - 0.1*exp(3*(x-1.6))}, {1.1*x - 2*exp(3*(x-1.6))});

	\end{axis}
\end{tikzpicture}
		\caption{\label{fig:modesandcones} Illustration of Theorem \ref{thm:verifycycle} in $\R^3$. The blue cone splits $\R^3$ into $\Qs_1$ and $\Qs_2$ the line is an invariant of $\Mm(1)$ and the plane is an invariant of $\Mm(2).$ Points indicate distinct sample trajectories $\{\xv_i\}$.}
	\end{center}
\end{figure}

To avoid a long detour in our exposition, we leave the proof to the appendix, instead providing here a depiction of the idea behind it: In Fig.~\ref{fig:modesandcones}, we have $m=1$ and $\Ys=\{1,2\}$, and the blue cone splits $\R^3$, the state space, in $\Qs_1$ and $\Qs_2$; the two plots have different matrices $\Mm(1)$. Runs $\{\xv_i\}$ that generate the trace $1^\omega$ are solutions of the linear system $\xv_{i+1} = \Mm(1)\xv_i$, one such example being depicted with white dots. Likewise, black dots show a run generating the trace $2^\omega$, and it has to be a solution of $\xv_{i+1} = \Mm(2)\xv_i$. In the example on the left, the black line is supported by one real eigenvector of $\Mm(1)$ and, as it belongs to $\Qs_1$, at least solutions on top of this eigendirection are runs of the PETC system $\Ss$. In our example, this eigenvector is associated with a dominant mode of $\Mm(1),$ so solutions starting close to it converge towards it. The plane depicted on the left of Fig.~\ref{fig:modesandcones} is an invariant of $\Mm(2)$ associated to complex conjugate eigenvalues. Solutions starting in this plane stay in this plane, spiraling towards the origin (in case the PETC implementation is stabilizing), which confirms that $2^\omega$ is also a behavior of $\Ss$. The example on the right shows the defective case where the converse does not hold: for that, assume that $\Qs_1$ does not include its depicted blue boundary; however, the black line representing an eigendirection of $\Mm(1)$ runs precisely on this boundary. In this example, the white dots represent a run $\{\xv_i\}$ in $\Qs_1$, thus generating the trace $1^\omega$, but no invariant of $\Mm(1)$ is a subset of $\Qs_1$. Because the depicted mode of $\Mm(1)$ is dominant, there are solutions that start close to its associated eigendirection that stay in $\Qs_1$ forever.

Based on Theorem \ref{thm:verifycycle}, in the non-defective cases we can verify a cyclic behavior $\sigma^\omega$ by taking the finitely many linear invariants $\As$ of $\Mm_\sigma$ and checking if $\As \setminus \{\O\} \subseteq \Qs_\sigma$, or, more explicitly, taking $\sigma = k_1k_2...k_m$,
	\begin{equation}\label{eq:As_sequence}
		\begin{aligned}
			\As \setminus \{\O\} &\subseteq \Qs_{k_1},\\
			\Mm(k_1)\As \setminus \{\O\} &\subseteq \Qs_{k_2},\\
			& \vdots \\
			\Mm(k_{m-1})...\Mm(k_1)\As \setminus \{\O\} &\subseteq \Qs_{k_m}.
		\end{aligned}
	\end{equation}
Because each $\Qs_k$ is an intersection of quadratic sets (see Eq.\ \eqref{eq:setq}), we must be able to check whether a linear space is a subset of a given quadratic set, which is nothing but a positive-(semi)definiteness check:
\begin{prop}[\cite{gleizer2021hscc}]\label{prop:subspaceposdef}
	Let $\As$ be a linear subspace with basis $\vv_1, \vv_2,..., \vv_m$, and let $\Vm$ be the matrix composed of the vectors $\vv_i$ as columns. Let $\Qm \in \S^n$ be a symmetric matrix and define $\Qs_n \coloneqq \{\xv \in \R^n \mid  \xv\tran\Qm\xv \geq 0\}$ and $\Qs_s \coloneqq \{\xv \in \R^n \mid  \xv\tran\Qm\xv > 0\}$. Then, $\As \setminus \{\O\} \subseteq \Qs_n$ (resp.~$\Qs_s$) if and only if $\Vm\tran\Qm\Vm \succeq \O$ (resp.~$\Vm\tran\Qm\Vm \succ \O$).
\end{prop}

\subsection{SACE simulation algorithm}

Combining the $l$-complete traffic models from §\ref{ssec:lcomplete} with the stopping criterion based on checking linear invariants from §\ref{ssec:verifycycle}, we specialize Algorithm \ref{alg:general} into Algorithm \ref{alg} to generate a finite-state SACE simulation of the PETC traffic model $\Ss$, together with the computation of its SAIST $\VS(\Ss)$. In the outer loop, the relation $\Rs_l$ and corresponding finite-state system $\Ss_l$ are built, followed by the computation of one of its SACs $\sigma$. Then, an inner loop looks for linear subspaces $\As$ of $\Mm_\sigma$ satisfying $\As \setminus \{\O\} \subseteq \Qs_\sigma$ (Theorem \ref{thm:verifycycle}); because $\Mm_\sigma$ is assumed to be mixed and of irrational rotations%
\footnote{Any matrix is arbitrarily close to a mixed matrix of irrational rotations; numerically checking if it is otherwise is not robust. A more thorough discussion about this is available in §\ref{ssec:robustness}}%
, it suffices to verify 1-dimensional subspaces for real eigenvectors and 2-dimensional subspaces for complex conjugate ones%
\footnote{If a larger dimensional subspace $\As'$ is a subset of $\Qs_\sigma$, any smaller dimensional subspace $\As \subset \As'$ will also be. Thus, there is no benefit in verifying subspaces that are combinations of smaller real linear subspaces}; if one is found, the algorithm terminates. Otherwise, $l$ is incremented and the main loop is repeated. Hereafter, we say that a linear invariant subspace of a mixed matrix is \emph{basic} if it is the span of a real eigenvector or of a pair of complex conjugate eigenvectors.
\begin{algorithm}\caption{\label{alg}PETC SAIST computation algorithm}
	\begin{flushleft}
		\hspace*{\algorithmicindent} \textbf{Input:} $\Ys$ and $\Mm(k), \Qs_k, \forall k \in \Ys$ \\
		\hspace*{\algorithmicindent} \textbf{Output:} $l, \Ss_l, \sigma, \mathtt{SAIST}$
	\end{flushleft}
	\begin{algorithmic}[1]
		\State $l \gets 1$
		\While{true}
		\State Build $\Rs_l$ and $\Ss_l$ \Comment{(Defs.~\ref{def:bisimrel} and \ref{def:lsim})}
		\State $\mathtt{SAIST} \gets \VS(\Ss_l), \ \sigma \gets \SAC(\Ss_l)$ \Comment{\cite{karp1978characterization, chaturvedi2017note}}
		\ForAll{$\As \in \mathrm{BILS}(\Mm_\sigma)$}  \Comment{BILS = basic invariant linear subspaces}	
		\If{$\As$ satisfies Eq.\ \eqref{eq:As_sequence} with $k_1,k_2,...,k_m = \sigma$}
		\State \Return
		\EndIf
		\EndFor
		\State $l \gets l+1$
		\EndWhile
	\end{algorithmic}
\end{algorithm}

In order to state formal results about the correctness of Algorithm \ref{alg}, we need to account for the conditions in Theorem \ref{thm:verifycycle}.

\begin{defn}[Normalized distance]\label{def:ndist}
	The \emph{normalized distance} between a point $\xv \in \R^n$ and a set $\As \subseteq \R^n$, denoted by $\nd(\xv,\As)$ is defined as $\inf_{\lv \in \As}\left(1 - \frac{\lv\tran\xv}{\norm[\lv]\norm[\xv]}\right)$. The normalized distance between two sets is $\nd(\As, \As') \coloneqq \inf_{\lv \in \As}\nd(\lv, \As').$	
\end{defn}

As the quantity $\frac{\lv\tran\xv}{\norm[\lv]\norm[\xv]}$ is the cosine of the angle between the vectors $\lv$ and $\xv$, the normalized distance varies between 0 and 1, measuring how close $\xv$ is, modulo magnitude, to the set $\As$. It is a more sensible choice of distance when dealing with homogeneous sets than the Euclidean distance, which would be zero as the origin is always in or arbitrarily close to such sets. This distance is needed for some technical results that come later, as well as for the following definition. 

\begin{defn}[Regularity]\label{def:regular}
	A sequence of ISTs $\sigma \coloneqq k_1k_2...k_m$ is said to be \emph{regular} if (i) $\Mm_\sigma$ is nonsingular, mixed, and of irrational rotations, and (ii) for every invariant linear subspace $\As$ of $\Mm_\sigma$, we have that 
	$\nd(\As, \partial\Qs_\sigma) \geq \epsilon$ for some $\epsilon > 0.$ 
\end{defn}

Regularity of a sequence $\sigma$ prevents that one of the invariants of $\Mm_\sigma$ intersect $\partial\Qs_\sigma$ (the case in the right of Fig.~\ref{fig:modesandcones}), requiring a minimal $\epsilon$ clearance to its boundary. %
The following result establishes conditions for the termination of Alg.~\ref{alg}; the proof is in the Appendix.
\begin{thm}\label{thm:algterminates}
	Suppose that $\Ss$ from Eq.\ \eqref{eq:S} has an isolated smallest-in-average cycle $\sigma$ that is regular. Then, Alg.~\ref{alg} terminates with $\mathtt{SAIST} = \VS(\Ss)$.
\end{thm}

The conditions of Theorem \ref{thm:algterminates} are the same behavioral conditions as in Theorem \ref{thm:Vcomputable}: the system must exhibit a minimizing periodic behavior, and competing infinite behaviors must be composed of subsequences that have average value strictly larger than the minimal value. Additionally, the smallest cycle must be regular, which is not a limiting assumption. Therefore, the algorithm may not terminate when, for example, a minimizing behavior is aperiodic. In this case, we may still expect increasingly better estimates of $\VS(\Ss)$ with larger values of $l$.

\subsection{Robustness and computability}\label{ssec:robustness}

Algorithm \ref{alg} relies on the matrices $\Mm(k)$ from Eq.\ \eqref{eq:petc_time}, whose elements are typically transcendental. Therefore, one may wonder if the algorithm, or more generically a given $l$-complete SACE traffic model, is robust to small round-off errors when computing those matrices, as well as other small model mismatches. In this section, we are going to see that this is true in the general case. For this, we need proper definitions.

\begin{defn}[Perturbed PETC system] \label{def:perturbed} Given a PETC system \eqref{eq:plant}--\eqref{eq:quadtrig} and its data $\Am, \Bm, \Km, \Qm, \bar{k}$, the PETC system with data $\tilde\Am, \tilde\Bm, \tilde\Km, \tilde\Qm, \bar{k}$ is called a $\delta$-perturbation of the former if $\norm[\Am - \tilde\Am] \leq \delta, \norm[\Bm\Km - \tilde\Bm\tilde\Km] \leq \delta,$ and $\norm[\Qm - \tilde\Qm] \leq \delta$. Furthermore, the traffic model $\tilde{\Ss}$ cf.~Eq.\ \eqref{eq:S} of a $\delta$-perturbation of system \eqref{eq:plant}--\eqref{eq:quadtrig} is denoted a $\delta$-perturbation of $\Ss$.
\end{defn}

\begin{rem}
	Considering footnote \ref{foot:hequals1}, Def.~\ref{def:perturbed} also encompasses variations in the actual checking period $h$.
\end{rem}

\begin{defn}[$\epsilon$-inflation]
	The \emph{$\epsilon$-inflation} of a quadratic cone $\{\xv \in \R^n \mid \xv\tran\Qm\xv \mathrel{\geq\!(>)} 0\}$ is the set $\{\xv \in \R^n \mid \xv\tran(\Qm +\epsilon\I)\xv \mathrel{\geq\!(>)} 0\}$, for $\epsilon > 0$. An $\epsilon$-inflation of the intersection of quadratic cones is defined as the intersection of the $\epsilon$-inflations.
\end{defn}

Let $\Ps_\delta(\Ss)$ be the set of all $\delta$-perturbations of $\Ss$. We have the following results.

\begin{prop}\label{prop:robust}
	Let $\Ss$, Eq.\ \eqref{eq:S}, be the traffic model of system \eqref{eq:plant}--\eqref{eq:quadtrig}. If $\Ss_l$ is an $l$-complete model thereof (Def.~\ref{def:lsim}), then there exists $\delta > 0$ such that $\Ss_l$ is an $l$-complete model of every $\tilde\Ss \in \Ps_\delta(\Ss)$ if there exists an $\epsilon > 0$ such that the following conditions hold:
	\begin{itemize}
		\item For every $\sigma \in \Bs^l(\Ss),$ there exists $\xv \in \Qs_\sigma$ s.t.~$\nd(\xv,\partial\Qs_\sigma) > \epsilon$; and
		\item for every $\sigma \notin \Bs^l(\Ss),$ every $\epsilon$-inflation of $\Qs_\sigma$ is empty.
	\end{itemize}
\end{prop}

\begin{proof}
	By Definition \ref{def:lsim}, $\Ss_l$ is an S$l$CA of every $\tilde\Ss \in \Ps_\delta(\Ss)$ if
	\begin{enumerate}
		\item $\sigma \in \Bs^l(\Ss) \implies \sigma \in \Bs^l(\tilde\Ss), \forall \tilde\Ss \in \Ps_\delta(\Ss),$ and
		\item $ \sigma \notin \Bs^l(\Ss) \implies \sigma \in \Bs^l(\tilde\Ss), \forall \tilde\Ss \notin \Ps_\delta(\Ss).$
	\end{enumerate}
	For item 1, we must have a non-zero vector $\xv \in \tilde{\Qs}_\sigma$, where $\tilde\Qs_\sigma$ is the $\sigma$-cone of the $\delta$-perturbation $\tilde\Ss$. Because $\nd(\xv,\partial\Qs_\sigma) > \epsilon$, we have that the normalized distance to the complement of $\Qs_\sigma$ satisfies $\nd(\xv,\bar\Qs_\sigma) > \epsilon$. By continuity, this implies that $\nd(\xv,\bar{\tilde\Qs}_\sigma) > 0$ for small enough $\delta$, and hence $\xv \in \tilde\Qs \implies \sigma \in \Bs^l(\tilde\Ss)$.
	Likewise, for item 2, we cannot have a vector $\xv \in \tilde{\Qs}_\sigma$; by continuity, for small enough $\delta$, $\tilde{\Qs}_\sigma$ is a subset of the $\epsilon$-inflation of $\Qs_\sigma$, which is empty, and therefore $\sigma \notin \Bs^l(\tilde\Ss)$.
\end{proof}

\begin{prop}\label{prop:robustcycle}
	Let $\sigma^\omega$ be a cyclic behavior of $\Ss$ from Eq.\ \eqref{eq:S}. Then, if $\sigma$ is regular, there exists some $\delta>0$ such that $\sigma^\omega \in \Bs^\omega(\tilde\Ss), $ for all $\tilde\Ss \in \Ps_\delta(\Ss)$.
\end{prop}

\begin{proof}
	From Theorem \ref{thm:verifycycle}, we have that $ \sigma^\omega \in \Bs^\omega(\Ss) \implies \As \subseteq \cl(\Qs_\sigma)$ for a basic linear invariant subspace $\As$ of $\Mm_\sigma$. From regularity of $\sigma$, $\nd(\As,\partial\Qs_\sigma) > \epsilon$. Together with $\As \subseteq \cl(\Qs_\sigma)$, we have that $\nd(\As,\bar\Qs_\sigma) > \epsilon.$ 
	Since $\sigma$ is regular, $\Mm_\sigma$ is mixed by definition. 
	Then, by continuity of eigenvalues and eigenvectors, for small enough $\delta$, the perturbed eigenvalues $\tilde\lambda_i$ are qualitatively unchanged: $\lambda_i \in \R \implies \tilde\lambda_i \in \R$, $\Im(\lambda_i) \neq 0 \implies \Im(\tilde\lambda_i) \neq 0$, and $\norm[\lambda_i] > \norm[\lambda_{i+1}] \implies |\tilde{\lambda}_i| > |\tilde\lambda_{i+1}|$. Thus, if $\As$ is a line associated to a real eigenvalue, so is the corresponsding basic linear subspace $\tilde\As$ of $\tilde\Mm_\sigma$; and likewise if $\As$ is a plane corresponding to complex conjugate eigenvalues of irrational rotations: even if $\tilde\Mm_\sigma$ is not of irrational rotations, the plane $\tilde\As$ is one of its invariants. In addition, $\nd(\As,\tilde\As) < d, $ where $d$ diminishes with $\delta$. Hence, for small enough $\delta$ we have that
	$ \nd(\As,\bar\Qs_\sigma) > \epsilon \implies \nd(\tilde\As,\bar{\tilde\Qs}_\sigma) > 0 \implies \tilde\As \setminus \{\O\} \subseteq \tilde\Qs_\sigma.$ Therefore, applying again Theorem \ref{thm:verifycycle}, we conclude that $\sigma^\omega \in \Bs^\omega(\tilde\Ss), \forall \tilde\Ss \in \Ps_\epsilon(\Ss).$
	
\end{proof}

These two propositions combined give the following result:

\begin{thm}\label{thm:robust}
	Let $\Ss$, Eq.\ \eqref{eq:S}, be the traffic model of system \eqref{eq:plant}--\eqref{eq:quadtrig}, and let $\Ss_l$ be its SACE simulation. If its smallest-in-average cycle $\sigma$ is regular and $\Ss_l$ satisfies the premises of Prop.~\ref{prop:robust}, then there exists $\delta > 0$ such that $\Ss_l$ is SACE simulation of every $\tilde\Ss \in \Ps_\delta(\Ss).$
\end{thm}

Theorem \ref{thm:robust} has two interesting implications. The first is that sufficiently small round-off errors on the matrices $\Mm(k)$ and $\Qm$ of Eq.~\eqref{eq:petc_time} do not affect the correct computation of $\VS(\Ss)$; hence, $\VS(\Ss)$ is computable for a class of linear systems, even though $\Mm(k)$ typically contains transcendental numbers. The second implication is that, informally, we can apply our method to nonlinear systems, as long as the closed loop ETC system is asymptotically stable and the involved functions are sufficiently smooth. Asymptotic stability implies that the state converges to a ball of any radius, no matter how small, in finite time; therefore, the sequence of sampling times up to this point do not affect the system's SAIST. Inside a sufficiently small ball, the nonlinear flow belongs to a convex combination of $\delta$-perturbations of its linearization about the equilibrium. If the linearized system $\Ss$ satisfies the premises of Theorem \ref{thm:robust}, the SAIST of the nonlinear system is equal to $\VS(\Ss)$.

\section{NUMERICAL EXAMPLES}\label{sec:numerical}

\subsection{A two-dimensional linear system}\label{ssec:num1}
	We start by considering the example from \cite{gleizer2021hscc}: the system \eqref{eq:plant} with
	$$ \Am = \begin{bmatrix}0 & 1 \\ -2 & 3 \end{bmatrix}, \quad \Bm = \begin{bmatrix}0 \\ 1\end{bmatrix}, \quad \Km =\begin{bmatrix}0 & -5\end{bmatrix}, $$
	and the triggering condition of \cite{tabuada2007event}, $|\xiv(t) - \hat{\xiv}(t)| > \sigma|\xiv(t)|$ for some $0 < \sigma < 1$, which can be put in the form Eq.\ \eqref{eq:quadtrig}. Checking time was set to $h=0.05$, and maximum inter-sample time to $\bar{k}=20$. Using a Python implementation of Algorithm \ref{alg} with Z3 \cite{demoura2008z3} to solve Eq.\ \eqref{eq:sequence}, we attempted to compute its SAIST through a SACE simulation for $\sigma \in \{0.1, 0.2, 0.3, 0.4, 0.5\}$. Table \ref{tab} presents the SAIST for each $\sigma$, as well as the $l$ value (Def.~\ref{def:lsim}) where it was obtained. Only for $\sigma = 0.1$ the algorithm did not terminate before $l=50$: for this case, the actual $\bar{k}$ of the system was 3, and all $\Mm(k), k \leq 3,$ have complex eigenvalues. Thus, it is possible that it does not have periodic behaviors, similarly to the irrational rotation of Example \ref{ex:irrational}. Nonetheless, applying Prop.~\ref{prop:upperbound} gives an upper bound for $\VS(\Ss)$ of 1.596; hence, we know that estimate $\VS(\Ss_l)$ is within only 0.024 of the real value. For the other cases, trivial cycles were found for $\sigma = 0.4~(5^\omega)$ and $\sigma = 0.5~(6^\omega)$, but it took a few iterations to break, e.g., the $2^\omega$ loop. Interestingly, the simplest cycles for $\sigma = 0.2$ and $\sigma = 0.3$ had length, respectively, 27 and 28, showing that PETC can often lead to very complex recurring patterns. In addition, the case of $\sigma=0.4$ has two verified cyclic behaviors, $5^\omega$ and $6^\omega$, while with $\sigma=0.5$ three cycles are obtained: $6^\omega, 7^\omega$ and $8^\omega$: this confirms that a single PETC system can exhibit multiple different periodic behaviors.
	
	The results were generated on a MacBook Pro 2017 using a single processor. As Table \ref{tab} shows, even for $l=50$ the CPU time was kept under 10 minutes.

\begin{table}\caption{\label{tab} SAIST values for the example of §\ref{ssec:num1}}
	\begin{center}
		\begin{tabular}{c|ccccc}
			\hline
			$\sigma$ & 0.1 & 0.2 & 0.3 & 0.4 & 0.5 \\
			\hline
			$l$ & 50* & 15 & 26 & 12 & 10 \\
			SAIST & 1.572 & 2.74 & 3.42 & 5 & 6 \\
			CPU time [s] & 327 & 41 & 147 & 29 & 45 \\
			\hline
		\end{tabular}
		
		{\footnotesize * Algorithm interrupted before finding a verified cycle.}
		\vspace{-1em}
	\end{center}
\end{table}

\subsection{A three-dimensional linear system}\label{ssec:num2}

With $\nx = 3$, the computational time involved in solving the existence problem of Eq.\ \eqref{eq:sequence} increased significantly. This is not surprising since solving such problems is exponential on the number of variables \cite{gleizer2020towards, basu1996combinatorial}. To reduce the number of times these problems are solved, we implemented a more efficient refinement approach than performing the full $(l+1)$-complete abstraction. At every iteration of Alg.~\ref{alg}, we only refine the states of the abstraction associated with the previous SAC. To illustrate this approach, see Fig.~\ref{fig:lincomplete}, where three steps of this refinement approach are executed: in depth 3, the SAC is already $(1,1,2)^\omega$, but it requires only 6 verifications: 1, 2, $(1,1), (1,2),$ $(1,1,1)$ (disproved) and $(1,1,2)$; the $3$-complete model would require up to $2+4+8 = 14$ verifications to obtain the same SAC. The disadvantage of this approach is that the obtained graph is more connected (as we have fewer states but more behaviors), and thus the computation of an upper bound using Prop.~\ref{prop:upperbound} often gives too distant values.

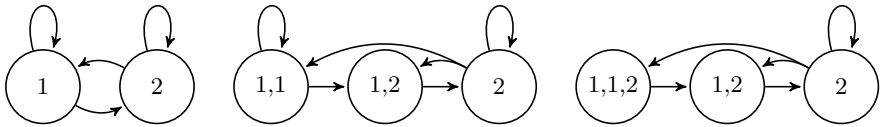
\begin{figure}[tb]
	\begin{center}
		\begin{footnotesize}
			\begin{tikzpicture}[->,>=stealth',shorten >=1pt, auto, node distance=1.5cm,
				semithick]
				\tikzset{every state/.style={minimum size=3em, inner sep=2pt}}
				
				\node[state] 		 (1)    {1};
				\node[state]         (2) [right of=1] {2};
				
				\path (1) edge[bend right] (2) (2) edge[bend right] (1)
				(1) edge [loop above] (1)
				(2) edge [loop above] (2);
				
				\node[state] 		 (11) [right of=2]  {1,1};
				\node[state]		 (12) [right of=11] {1,2};
				\node[state]		 (22) [right of=12] {2};
				
				\path (11) edge (12) (11) edge [loop above] (11) (12) edge (22) (22) edge[bend right] (11) (22) edge[bend right] (12) (22) edge [loop above] (22);
				
				\node[state] 		 (112) [right of=22]  {1,1,2};
				\node[state]		 (122) [right of=112] {1,2};
				\node[state]		 (222) [right of=122] {2};
				
				\path (112) edge (122) (122) edge (222) (222) edge[bend right] (112) (222) edge[bend right] (122) (222) edge [loop above] (222);

			\end{tikzpicture}
		\end{footnotesize}
		\caption{\label{fig:lincomplete} Illustration of the specialized refinement used in place of full $(l+1)$-complete abstraction in Alg.~\ref{alg}.}
		\vspace{-1.5em}
	\end{center}
\end{figure}

We applied this improved version of Alg.~\ref{alg} to system \eqref{eq:plant}--\eqref{eq:quadtrig} with
$$ \Am = \begin{bmatrix}0 & 1 & 0 \\ 0 & 0 & 1 \\ 1 & -1 & -1 \end{bmatrix}, \quad \Bm = \begin{bmatrix}0 \\ 0 \\ 1\end{bmatrix}, \Km = \begin{bmatrix}-2 & -1 & -1\end{bmatrix}, $$
with $h = 0.1, \bar{k} = 20$ and the triggering condition $|\xiv(t) - \hat{\xiv}(t)| > \sigma|\xiv(t)|.$ This time, some parallelization was also applied: at most 10 threads of an Intel\textregistered~Xeon\textregistered~W-2145 CPU were used, solving multiple instances of Eq.\ \eqref{eq:sequence} in parallel whenever possible. Table \ref{tab2} shows the results for multiple choices of $\sigma$, where $l$ now is the largest length of any state in the abstraction. Only for $\sigma=0.2$ the algorithm was interrupted without finding an exact value. The CPU times vary dramatically, in some cases taking less than a minute, whilst in others reaching an hour. The most interesting thing we observe is that, even though the SAIST never decreases with $\sigma$ as expected, there is not a consistent increase on its values after $\sigma=0.3$. This is reasonable considering the results of §\ref{ssec:robustness}: for small enough perturbations of the ETC system's parameters, the same cycle may still be present (Prop.~\ref{prop:robustcycle}). Interestingly, for $\sigma=0.9$ there is a substantial jump in the SAIST value.

\begin{table}\caption{\label{tab2} SAIST values for the example of §\ref{ssec:num2}}
	\begin{center}
		\begin{tabular}{c|ccccccccc}
			\hline
			$\sigma$ & 0.1 & 0.2 & 0.3 & 0.4 & 0.5 & 0.6 & 0.7 & 0.8 & 0.9 \\
			\hline
			$l$ & 1 & 18* & 14 & 8 & 6 & 7 & 6 & 5 & 9 \\
			SAIST & 1 & 1.921 & 3 & 3 & 3 & 4 & 4 & 4 & 9.5 \\
			CPU time [s] & 2 & 3056 & 1551 & 95 & 185 & 236 & 153 & 40 & 2955 \\
			\hline
		\end{tabular}
	\end{center}
\end{table}

\subsection{A nonlinear system}\label{ssec:num3}

Consider now the PETC triggering rule $|\xiv(t) - \hat{\xiv}(t)| > \sigma|\xiv(t)|$ with $h=0.05, \sigma = 0.452 $ applied to the following nonlinear jet engine system \cite{delimpaltadakis2020isochronous}:
\begin{align*}
	\dot\xi_1(t) &= -\xi_2(t) - 1.5\xi_1(t)^2 - 0.5\xi_1(t)^3 \\
	\dot\xi_2(t) &= \upsilon(t), \\
	\upsilon(t) & = \hat\xi_1(t) - 0.5(\hat\xi_1(t)^2 + 1)(y(t) + \hat\xi_1(t)^2y(t) + \hat\xi_1(t)y(t)^2),
\end{align*}
where $y(t) = (\hat\xi_1(t)^2 + \hat\xi_2(t))/(\hat\xi_1(t)^2 + 1)$. The origin of the closed-loop system is asymptotically stable%
\footnote{For stability analysis of PETC of nonlinear systems, see, e.g.,~\cite{postoyan2013periodic}.}%
, therefore we can obtain its SAIST through its linearized model around the origin, which is of the form \eqref{eq:plant} with
$$ \Am = \begin{bmatrix}0 & -1 \\ 0 & 0 \end{bmatrix}, \quad \Bm = \begin{bmatrix}0 \\ 1\end{bmatrix}, \quad \Km =\begin{bmatrix}1 & -0.5\end{bmatrix}. $$
We ran Alg.~\ref{alg} and stopped it with $l=100,$ obtaining an approximate value of $\VS(\Ss) = 8.882$. Using Prop.~\ref{prop:upperbound}, an upper bound of 8.892 was obtained, thus giving an error of 0.01. Figure~\ref{fig:nonlinear} shows ISTs and their running averages for five PETC simulations starting each from a different pseudo-randomly generated initial state, for both the nonlinear model and the linearized model. It can be seen that the running averages in both cases converge to the predicted SAIST value, even though the averages are significantly different in the beginning of the simulation. The right plot shows how the difference between ISTs based on the nonlinear model and the linear model diminish as the state norm approaches zero: in the plotted simulation there is no error after the state norm is below 0.03 (around time instant 400).

\begin{figure}[tb]
	\begin{center}
		\begin{footnotesize}
			\input{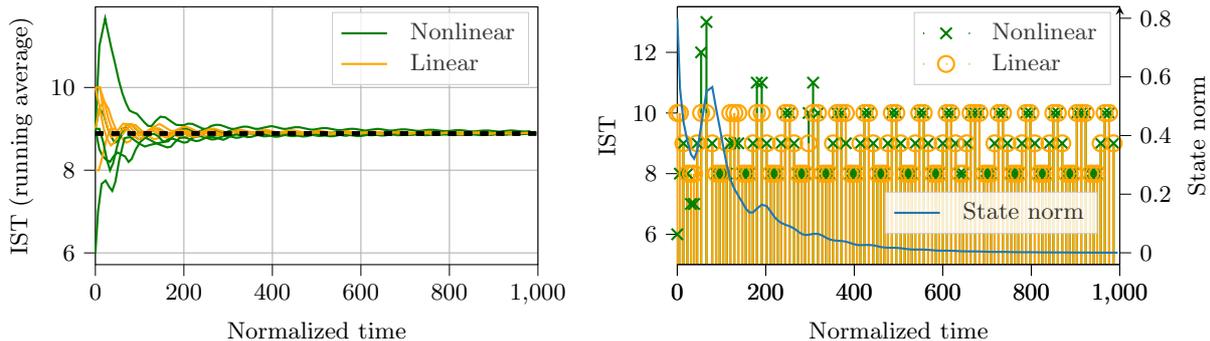} 
\begin{tikzpicture}

\definecolor{color0}{rgb}{1,0.647058823529412,0}
\definecolor{color1}{rgb}{0.12156862745098,0.466666666666667,0.705882352941177}

\begin{axis}[
height=5cm,
width=0.45\linewidth,
legend cell align={left},
legend style={
  fill opacity=0.8,
  draw opacity=1,
  text opacity=1,
  draw=white!80!black
},
tick align=outside,
tick pos=left,
x grid style={white!69.0196078431373!black},
xlabel={Normalized time},
xmin=0, xmax=1000,
xtick style={color=black},
y grid style={white!69.0196078431373!black},
ylabel={IST},
ymin=5, ymax=13.5,
ytick style={color=black}
]

\addplot [ycomb, thick, green!50.1960784313725!black, mark=x, mark size=3, mark options={solid}]
table {%
0 6
6 8
14 9
23 8
31 7
38 7
45 9
54 12
66 13
79 9
88 7.99999999999999
96 8.00000000000001
104 8
112 8
120 9
129 9
138 9
147 8
155 8
163 8
171 9
180 11
191 11
202 9.00000000000003
211 8
219 7.99999999999997
227 8.00000000000003
235 9
244 10
254 10
264 9
273 8
281 8
289 8
297 10
307 11
318 9.99999999999994
328 8.00000000000006
336 7.99999999999994
344 8.00000000000006
352 9
361 10
371 10
381 9
390 8
398 8
406 7.99999999999994
414 9.00000000000006
423 10
433 10
443 9
452 8
460 8
468 8
476 9
485 10
495 10
505 9
514 8
522 8
530 8
538 9
547 10
557 10
567 9
576 8
584 8
592 8
600 9
609 10
619 10
629 9
638 7.99999999999989
646 8.00000000000011
654 8.99999999999989
663 10
673 10
683 10
693 8.00000000000011
701 8
709 8
717 9
726 10
736 10
746 9
755 7.99999999999989
763 8.00000000000011
771 8
779 8.99999999999989
788 10
798 10
808 9.00000000000011
817 8
825 7.99999999999989
833 8.00000000000011
841 9
850 10
860 10
870 9
879 8
887 8
895 7.99999999999989
903 10
913 10
923 10
933 8.00000000000011
941 8
949 8
957 9
966 10
976 10
986 9
};
\addlegendentry{Nonlinear};
\addplot [ycomb, thick, color0, mark=o, mark size=3, mark options={solid}]
table {%
0 10
6 10
14 9
23 8
31 8
38 8
45 9
54 10
66 10
79 9
88 8
96 8
104 8
112 9
120 10
129 10
138 10
147 8
155 8
163 8
171 9
180 10
191 10
202 9
211 8
219 8
227 8
235 9
244 10
254 10
264 9
273 8
281 8
289 8
297 9
307 10
318 10
328 8
336 8
344 8
352 9
361 10
371 10
381 10
390 8
398 8
406 8
414 9
423 10
433 10
443 9
452 8
460 8
468 8
476 9
485 10
495 10
505 9
514 8
522 8
530 8
538 9
547 10
557 10
567 9
576 8
584 8
592 8
600 9
609 10
619 10
629 9
638 8
646 8
654 9
663 10
673 10
683 10
693 8
701 8
709 8
717 9
726 10
736 10
746 9
755 8
763 8
771 8
779 9
788 10
798 10
808 9
817 8
825 8
833 8
841 9
850 10
860 10
870 9
879 8
887 8
895 8
903 10
913 10
923 10
933 8
941 8
949 8
957 9
966 10
976 10
986 9
};
\addlegendentry{Linear};
\end{axis}

\begin{axis}[
height=5cm,
width=0.45\linewidth,
axis y line=right,
legend cell align={left},
legend style={fill opacity=0.8, draw opacity=1, text opacity=1, draw=white!80!black,
			  at={(0.95, 0.2)}, anchor=east},
tick align=outside,
x grid style={white!69.0196078431373!black},
xmin=0, xmax=1000,
xtick pos=left,
xtick style={color=black},
y grid style={white!69.0196078431373!black},
ylabel={State norm},
ymin=-0.0394966078499901, ymax=0.839678475892103,
ytick pos=right,
ytick style={color=black},
yticklabel style={anchor=west}
]
\addplot [thick, color1]
table {%
0 0.799715972085644
6 0.56137126536259
14 0.445988916854086
23 0.379394578665527
31 0.335243245318712
38 0.320792920354273
45 0.348226461915074
54 0.434251921397413
66 0.550658617059709
79 0.565315395104389
88 0.489572902735446
96 0.426384713729292
104 0.360634417594361
112 0.297909555705774
120 0.250666962145365
129 0.217718856340001
138 0.194534422239243
147 0.17181185480386
155 0.150714742425054
163 0.137327807245969
171 0.137954228769987
180 0.150802221748542
191 0.16415106623528
202 0.159219759263973
211 0.14080019950275
219 0.123425980495446
227 0.110424875823115
235 0.101896174937227
244 0.0958306108216987
254 0.0905903818589432
264 0.0828479521531039
273 0.0727096313194024
281 0.0643924935536556
289 0.0606172103738191
297 0.0615482076061661
307 0.0645967426544526
318 0.0645379503744039
328 0.0586802117439493
336 0.0514739703027668
344 0.0457653318650478
352 0.0425981968368932
361 0.0412464334340373
371 0.0403738321858476
381 0.0380440541382531
390 0.0339228334476649
398 0.0297605702663013
406 0.026927178019427
414 0.0260918763875004
423 0.0265739774701049
433 0.0268375988422364
443 0.0253840758497798
452 0.0225229282257624
460 0.019772647007514
468 0.0179471213394205
476 0.0172388582878066
485 0.0171163087494311
495 0.0168205984259743
505 0.0156308104155643
514 0.0137692144372032
522 0.0121385923575411
530 0.0112151362887221
538 0.0110407750550528
547 0.0111766626900882
557 0.0110345937351333
567 0.0101897940933116
576 0.00893601475924074
584 0.00790940366202437
592 0.00735702514338621
600 0.00723042784872215
609 0.00723275177386228
619 0.00703525669070098
629 0.00642216259378011
638 0.00561185448923776
646 0.00500405014700334
654 0.00472801400561391
663 0.00471531980024352
673 0.00474214405173238
683 0.00452664028818171
693 0.0040466115287525
701 0.00355013266889641
709 0.00321172768148624
717 0.00308712565676071
726 0.00308879732762945
736 0.00306622621690779
746 0.00287213536211081
755 0.00253934402114091
763 0.00223303020496597
771 0.00204467271823159
779 0.00199194447832322
788 0.00200423809318229
798 0.00198038851763691
808 0.00183736806863383
817 0.00161539569081626
825 0.00142608630227266
833 0.0013201224366142
841 0.00129594462137388
850 0.00130168716520233
860 0.00127457456484994
870 0.00117037657319327
879 0.00102454006945979
887 0.000909731323282969
895 0.000852699554678687
903 0.000844493156083245
913 0.000850616148065414
923 0.00082091144169696
933 0.00074165084969263
941 0.000650526251253887
949 0.000582214815101001
957 0.000552519819314401
966 0.000550856991840358
976 0.000551090967127924
986 0.00052267902828491
995 0.0004658959564687
};
\addlegendentry{State norm}
\end{axis}

\end{tikzpicture} 
		\end{footnotesize} 
		\vspace{-1em}
		\caption{\label{fig:nonlinear} Left: running average of ISTs of five nonlinear PETC simulations and of five corresponding linear PETC simulations, with the dashed black line representing the estimated SAIST. Right: ISTs for one nonlinear PETC simulation and the corresponding ISTs predicted by the linear PETC model, with the state norm  overlaid on a secondary axis.}
		\vspace{-1.5em}
	\end{center}
\end{figure}
\section{CONCLUSIONS}\label{sec:conc}

We have presented a method to compute the sampling performance of PETC, namely its minimum average inter-sample time, by means of an abstraction called SACE simulation. For this we rely on methods of abstracting and refining to obtain tighter simulations, and getting their smallest-in-average cycle through Karp's algorithm. A SACE simulation requires that this cycle, repeated \emph{ad infinitum}, is a behavior of the concrete system; for this, we need to find an invariant of the system, which is possible for PETC of linear systems through the inspection of linear invariants of an associated discrete-time linear system. In the generic case --- quotient sets with non-empty interior and linear invariants not touching the boundary of the cones they belong to --- a SACE simulation is proven to be robust to small model uncertainties, which allows us to use the presented method to a large class of nonlinear systems. Even if an exact SACE simulation is not obtained, every simulation provides a lower bound to the SAIST, and upper bounds can also be computed from the abstractions. Our numerical results indicate that these bounds can be very close after sufficient refinements.

As with most applications of finite-state abstractions, our approach suffers from the ``curse of dimensionality'': with a three-dimensional system the computation can reach nearly an hour to complete. In fact, it can be argued that this curse is more severe in our case than in most control and verification applications, since we rely on \emph{strongest} $l$-complete abstractions, which require no spurious behavior fragments of length up to $l$. This may prevent the usage of most reachability tools to this end, as over- or under-approximations can create such spurious behaviors or remove potentially important ones. This is one of the reasons why we have used Z3 for our implementation, as it is one of the few exact nonlinear SAT solvers available. Nevertheless, the robustness results we have presented indicate that exactness may not be necessary in most cases. With this in mind, we plan to use approximate nonlinear SMT solvers such as dReal \cite{gao2013dreal} to start addressing the issue of dimensionality.

It is interesting to observe that the problem of computing the (smallest) limit average metric of an infinite system is highly dependent on its infinite behavior properties: systems with aperiodic behavior can make it impossible to obtain a SACE simulation, but other pathological behaviors can be even worse, such as the infamous ${(1^n2^n)^\omega}$, where not even a good approximation can be achieved. Better behavioral understanding of systems is crucial for the further development of quantitative verification methods. Part of this behavioral understanding of ETC sampling is currently the subject of our investigation.

Finally, natural extensions of this line of work are ongoing, such as extending it to systems with disturbances, in particular stochastic noise \cite{delimpaltadakis2021abstracting}, as well as the usage of abstractions for \emph{synthesis} of sampling strategies that maximize the closed-loop SAIST \cite{gleizer2021cdc}.

\bibliographystyle{elsarticle-num} 
\bibliography{mybib}

\appendix

\section{Proof of Theorem \ref{thm:Vcomputable}}

The proof relies on the notion of \emph{cyclic permutations}. A word $\sigma'$ is called a cyclic permutation of $\sigma \coloneqq a_0a_1...a_n$ if $\sigma' = a_ia_{i+1}...a_na_0a_1...a_{i-1}$ for some $i \leq n$. For example, the cyclic permutations of 1234 are 1234, 2341, 3412, and 4123. Clearly, all $n$-long subsequences of $\sigma^\omega$ are precisely the cyclic permutations of $\sigma$. Now we introduce the following Lemmas:

\begin{lem}\label{lem:cyclic1}
	Let $\sigma \in \Ys^n$ and $\sigma' \in \Ys^n$ be cyclic permutations of each other. If $\sigma = \alpha a$ and $\sigma' = \alpha b$, where $\alpha \in \Ys^{n-1}$ and $a, b \in \Ys$, then $a = b$ and thus $\sigma = \sigma'$.
\end{lem}

\begin{proof}
	Let $\sigma = a_0a_1...a_{n-1}$. Then $\sigma' = a_ia_{i+1}...a_{n-1}a_0...a_{i-1}$ for some $i > 0$ (if $i = 0$ the result is trivial). If their $(n-1)$-long prefixes are equal, then $a_j = a_{j + i \bmod n}$ for all $j < n - 1$. In particular, take $j = i - 1$; then $a_{i-1} = a_{2i - 1 \bmod n} = a_{3i - 1 \bmod n} = ... = a_{ki - 1 \bmod n}$, where $k$ is the smallest number such that $ki - 1 \bmod n = n-1$ (in the worst case, $k = n,$ for $i$ and $n$ coprime). Thus, $a_{i-1} = a_{ki - 1 \bmod n} = a_{n-1},$ concluding the proof.   
\end{proof}

\begin{lem}\label{lem:cyclic2}
	Let $\sigma \in \Ys^n$ and $\sigma' \in \Ys^n$ be cyclic permutations of each other. If $\sigma \neq \sigma',$ then there is a subsequence $\alpha$ of length $n$ of $\sigma\sigma'$ that is not a cyclic permutation of $\sigma$.
\end{lem}

\begin{proof}
	Let $\sigma = a_0a_1...a_{n-1}$. Then $\sigma' = a_ia_{i+1}...a_{n-1}a_0...a_{i-1}$ for some $i > 0$. We have $\sigma\sigma' = a_0a_1...a_{n-1}a_ia_{i+1}...a_{n-1}a_0...a_{i-1}$.
	
	Suppose, for contradiction, that every $n$-long subsequence of $\sigma\sigma'$ is a cyclic permutation of $\sigma$. Let us look at the first nontrivial subsequence, $\sigma_1 \coloneqq a_1...a_{n-1}a_i$. Because $a_1...a_{n-1}a_0$ is a cyclic permutation of $\sigma$, from Lemma \ref{lem:cyclic1} we get that $a_0 = a_i$. 
	Now let us apply induction: suppose that for some $J < n$, $a_j = a_{i+j \bmod n}$ for all $j < J$; we are going to show that this also holds for $j = J$. %
	First, suppose that $J < n - i$; then $\sigma_J = a_Ja_{J+1}...a_{n-1}a_ia_{i+1}...a_{i+J-2}a_{i+J-1} = a_Ja_{J+1}...a_{n-1}a_0a_{1}...a_{J-2}a_{i+J-1}.$ Again, because $a_Ja_{J+1}...a_{n-1}a_0a_{1}...a_{J-1}$ is a cyclic permutation of $\sigma$, apply Lemma \ref{lem:cyclic1} to obtain $a_{i+J-1} = a_{J-1}.$ %
	Second, suppose that $J \geq n - i$. Then, $\sigma_J = a_J...a_{n-1}a_i..a_{n-1}a_0a_1...a_{i+J-n-1} = a_J...a_{n-1}a_0..a_{n-i-1}a_0a_1...a_{i+J-n-1}.$ Note that $a_k = a_{k+n \bmod n} = a_{k+n-i}$ as long as $k+n-i<J,$ i.e., $k<i+J-n$. Thus, $\sigma_J =  a_J...a_{n-1}a_0..a_{n-i-1}a_{n-i}...a_{J-2}a_{i+J-n-1}.$ Again, apply Lemma \ref{lem:cyclic1} to get that $a_{i+J-n-1} = a_{J-1}$. We have that $J-1+i \bmod n = i+J-n-1$, since $n > J \geq n - i$; our hypothesis is thus confirmed. %
	The fact that $a_j = a_{i+j \bmod n}$ for all $j < n$ implies that $\sigma' = a_ia_{i+1}...a_{n-1}a_0...a_{i-1} = a_0a_1...a_{n-1-i}a_{n-i}...a_{n-1} = \sigma$, which contradicts the fact that $\sigma \neq \sigma'$.
\end{proof}

\begin{proof}[Proof of Theorem \ref{thm:Vcomputable}]
	From Theorem \ref{thm:lcompleteisoptimal}, there is an $l$ large enough such that $\VS(\Ss_l) = \VS(\Ss)$. It is easy to see that taking $l \geq m$ ensures that $\sigma$ is a cycle of the graph associated to $\Ss_l$.
	
	We prove that, because now $\LimAvg(\beta^\omega) > \VS(\Ss)$ for every $\beta$ that is not a subsequence of $\sigma^\omega$ (thus not a cyclic permutation of $\sigma$), the SAC of $\Ss_l$ is unique up to cyclic permutations. Suppose, for contradiction, that another cycle $\sigma'$ is a SAC of $\Ss_l$, with $|\sigma'| = p$. As in the proof of Theorem \ref{thm:lcompleteisoptimal}, we divide $(\sigma')^m$ into $p$ subsequences of length $m$, obtaining
	$$ \VS(\Ss_l) = \LimAvg((\sigma')^\omega) = \LimAvg(((\sigma')^m)^\omega)) = \Avg((\sigma')^m) = \frac{1}{p}\sum_{i=1}^p\Avg(\beta_i). $$
	If (i) some $\beta_i$ is not a cyclic permutation of $\sigma$, $\frac{1}{p}\sum_{i=1}^p\Avg(\beta_i) > \VS(\Ss_l)$, which yields the contradiction. Now, suppose (ii) that every $\beta_i$ is a cyclic permutation of $\sigma$; since $\sigma'$ is not the same cycle as $\sigma$, it cannot be that $\beta_i = \beta_j$ for all $i, j \leq p$. If $\beta_i \neq \beta_j$ for some $i, j$, suppose without loss of generality that they are adjacent in $(\sigma')^\omega$, i.e., either $j = i+1$ or $i=p$ and $j=1$. Then we have from Lemma \ref{lem:cyclic2} that there exists an $m$-long subsequence of $\beta_i\beta_j$ that is not a cyclic permutation of $\sigma$. Thus, $\sigma'$ has at least one subsequence $\beta'$ with average larger than $\VS(\Ss),$ which brings us back to case (i). The contradiction is thus achieved in all cases.
	
	Concluding, $\Ss_l$ has only one cycle $\sigma$ (modulo cyclic permutations) that attains its minimum value. Hence, running Karp's algorithm (Theorem \ref{thm:limavg}) retrieves it; by assumption, $\sigma^\omega \in \Bs^\omega(\Ss)$, thus the algorithm terminates at line 6.
\end{proof}

\section{Proof of Theorem \ref{thm:verifycycle}}

Before the main proof, we need some definitions. Given a map $f : X \to X$ and the discrete-time autonomous system defined by $x_{i+1} = f(x_i),$ we call the \emph{forward orbit} of $x$ the set $\bigO(x) \coloneqq \{f^n(x) \mid n \in \N\}$. The $\omega$-limit set of $x$, denoted by $\omega(x)$ is the set of cluster points of $\bigO(x)$, or alternatively, 
$$ \omega(x) = \bigcap_{n\in\N} \cl(\{f^k(x) \mid k > n\}). $$
By definition of closure, if $\bigO(x) \subset \As \subset X,$ then $\omega(x) \subset \cl(\As).$

We introduce the following Lemma.
\begin{lem}\label{lem:subspacesuffices}
	Let $\Mm \in \R^{n \times n}$ be a nonsingular mixed matrix and $\Qs \subseteq \R^n$ be a homogeneous set, i.e., it satisfies $\xv \in \Qs \implies \lambda\xv \in \Qs, \forall \lambda \in \R \setminus \{0\}$.
	If there exists a trajectory $\xiv: \N \to \R^\nx$ satisfying $\xiv(k+1) = \Mm\xiv(k)$ and $\xiv(k) \in \Qs$ $\forall k \in \N$, then there exists a linear subspace $\As$ that is an invariant of $\Mm^q$ and satisfies $\As \subseteq \cl(\Qs)$, where $q \in \N$. Furthermore, $q=1$ if $\Mm$ is of irrational rotations.
\end{lem}
\begin{proof} Because $\Qs$ is homogeneous, $\xiv(k) \in \Qs$ for all $k$ implies that the normalized trajectory $\xiv(k)/|\xiv(k)| \in \Qs$ for all $k$; likewise, for any constant $c \neq 0$, we have that $c\xiv(k)/|\xiv(k)| \in \Qs$. 
Therefore, let us investigate the ``normalized'' version of the iteration $\xv_{i+1} = \Mm\xv_i$: this is defined by the map $f : B^n \to B^n,$ where $B^n$ is the unit ball in $\R^n$ and $f(\xv) = \Mm\xv/\norm[\Mm\xv]$. Our strategy is to first determine what is $\omega(\xv)$; then, we will prove that the set $\{c\omega(\xv) \mid c \in \R \setminus \{0\}\}$, a radial expansion of $\omega(\xv)$, is a linear subspace of $\Mm$. Because $\omega(\xv) \subseteq \cl(\Qs)$ and $\xv \in \Qs \implies c\xv \in \Qs$, we conclude that $\{c\omega(\xv) \mid c \in \R \setminus \{0\}\} \subseteq \cl(\Qs)$. 

Now we investigate case by case depending on the eigenvalues $\lambda_i$ of $\Mm$. Since $\Mm$ is mixed, it is diagonalizable, and hence the trajectory $\xiv(k)$ can be decomposed as $\sum_{i=1}^n a_i\vv_i\lambda_i^k$, where $\vv_i$ are the eigenvectors of $\Mm$ satisfying $|\vv_i| = 1$. Let $m \leq n$ such that $a_i = 0$ for $i<m$, hence $\lambda_m$ is the dominant eigenvalue for this initial condition. Throughout, let $\xv \coloneqq \xiv(0)/\norm[\xiv(0)]$.
	
	{\bfseries Case 1:} $\lambda_m$ is real. Then
	\begin{multline*}
		\lim_{k\to\infty}\frac{\xiv(k)}{|\xiv(k)|} = \lim_{k\to\infty}\frac{a_{m}\vv_{m}\lambda_{m}^k + ...+ a_n\vv_n\lambda_n^k}{|a_{m}\vv_{m}\lambda_{m}^k + ...+ a_n\vv_n\lambda_n^k|}\\
		= \lim_{k\to\infty}\frac{a_{m}\vv_{m} + ...+ a_n\vv_n\left(\frac{\lambda_n}{\lambda_m}\right)^k}{\norm[a_{m}\vv_{m} + ...+ a_n\vv_n\left(\frac{\lambda_n}{\lambda_m}\right)^k]}
		= \lim_{k\to\infty}\frac{a_{m}\vv_{m}}{\norm[a_{m}\vv_{m}]} = \pm a_m\vv_m.
	\end{multline*}
	Hence, the set $\{c\omega(\xv) \mid c  \in \R \setminus \{0\}\}$ is the line $\{\pm c\vv_m \mid c \in \R \setminus \{0\}\} = \{c\vv_m \mid c \in \R \setminus \{0\}\},$ which is an invariant of $\Mm$.
	
	For the next cases, $\lambda_m$ and $\lambda_{m+1}$ form a complex conjugate pair, thus $\vv_{i+1} = \vv_i^*$. Denote by $\theta \coloneqq \arg \lambda_m$.
	
	{\bfseries Case 2:} $\theta / \uppi \notin \Q$. Using a similar approach as Case 1, we get $\lim_{k\to\infty}\xiv(k)/|\xiv(k)| = \pm(\vv_m\e^{\imag\theta k} + \vv_{m+1}\e^{-\imag\theta k}).$ Because $\theta$ is not a rational multiple of $\pi$, $\{k\theta \mid k \in \N\}$ is a dense subset of $[0, 2\piconst]$ and, therefore, $\omega(\xv) = \cl(\{\pm(\vv_m\e^{\imag\theta k} + \vv_{m+1}\e^{-\imag\theta k} \mid \theta \in k \in \N\}$ which is equal to the ellipse $\Bs \coloneqq \{\vv_m\e^{\imag\alpha} + \vv_{m+1}\e^{-\imag\alpha} \mid \alpha \in [0, 2\pi)\}.$ The set $\{c\xv \mid \xv \in \Bs, c \in \R \setminus \{0\}\}$ is the unique plane supported by $\vv_m$ and $\vv_{m+1}$, and as such is an invariant of $\Mm$.
	
	{\bfseries Case 3:} $\theta / \uppi = p/q,$ where $p, q \in \N$ are co-prime. The $m$-th and $(m+1)$-th eigenvalues of $\Mm$ have the form $r\e^{\pm \imag p\uppi/q}$, and as a consequence the corresponding eigenvalues of $\Mm^q$ are $\lambda_m^q = \lambda_{m+1}^q = r^q \in \R$. The geometric multiplicity of $\lambda_m^q$ is 2, since $\Mm^q$ is also diagonalizable. Thus, we have that $\lim_{k\to\infty}\xiv(qk)/|\xiv(qk)| = a_m\vv_i + a_{m+1}\vv_{i+1} \eqqcolon \wv.$ Hence, we have $\omega(\xv) \supseteq \{c\wv \mid c \in \R \setminus \{0\}\}$, a line that is an invariant of $\Mm^q$. Finally, this line is a subset of $\cl(\Qs)$, since $\{c\wv \mid c \in \R \setminus \{0\}\} \subseteq \omega(\xv) \subseteq \cl(\Qs)$.
\end{proof}
\begin{proof}[Proof of Theorem \ref{thm:verifycycle}]
	Statement (i), $\As \setminus \{\O\} \subseteq \Qs_\sigma$ implies $\sigma^\omega \in \Bs(\Ss)$, is straightforward. Take any point $\xv \in \As \subseteq \Qs_\sigma$. By definition of $\Qs_\sigma$, we have that $\xv \in \Qs_{k_1}, \Mm(k_1)\xv \in \Qs_{k_2}, ...,$ and $\Mm(k_{m-1})\cdots\Mm(k_1)\xv \in \Qs_{k_m}.$ The $(m+1)$-th element of the run starting from initial state $\xv$ is $\xv' = \Mm(k_m)\Mm(k_{m-1})\cdots\Mm(k_1)\xv = \Mm_\sigma\xv.$ Since $\As$ is an invariant of $\Mm_\sigma$ and this matrix is nonsingular, $\xv' \in \As \setminus \{\O\}.$ Thus, the behavior from $\xv$ is $\sigma\Bs_{\xv'}(\Ss)$. Applying the same reasoning recursively with $\xv'$ in the place of $\xv$, we conclude that $\Bs_{\xv}(\Ss) = \sigma^\omega.$
	
	Statement (ii) follows from Lemma \ref{lem:subspacesuffices}, by applying it with $\Qs = \Qs_\sigma$ and $\Mm = \Mm_\sigma$, and using the fact that $\Qs_\sigma$ is an homogeneous set.
	
\end{proof}

\end{document}